\documentclass[a4paper]{amsart}
\usepackage[margin=2cm]{geometry}
\usepackage[super,authoryear,longnamesfirst]{natbib}

\usepackage{caption}
\usepackage{subcaption}
\usepackage{graphicx}
\usepackage{float}
\usepackage{enumerate}
\usepackage[cyr]{aeguill}
\usepackage{array}             
\usepackage{amsmath,amsthm,amssymb,amsfonts} 
\usepackage{url}
\usepackage[all]{xy}
\usepackage{todonotes}
\usepackage{bm}
\usepackage{nicefrac}
\usepackage[inline]{asymptote}
\usepackage{makecell}
\usepackage{multirow}
\usepackage{pbox}
\usepackage{epsfig}
\usepackage{color}
\usepackage[table]{xcolor}
\usepackage{indentfirst}
\usepackage{autobreak}
\usepackage{dsfont}
\usepackage{tabularht}
\usepackage{threeparttable}
\usepackage{appendix}
\usepackage{longtable}
\usepackage{tikz-cd}
\usepackage{hyperref}
\usepackage[capitalise,noabbrev]{cleveref}

\usepackage{placeins}
\usepackage[english]{babel}

\DeclareMathAlphabet\mathbfcal{OMS}{cmsy}{b}{n}

\usepackage{booktabs}
\usepackage{multicol}    
\usepackage{supertabular}  
\usepackage{lipsum}
\usepackage{ltablex}      
\usepackage{float}
\newcommand{\RR}{\mathbb{R}} 
   
\newcommand{\Ela}{\mathbb{E}\mathrm{la}}

\newcommand{\SymC}{\mathfrak{I}}

\newcommand{\tr}{\mathrm{tr}}

\newcommand{\Fix}{\mathbb{F}\mathrm{ix}}
\newcommand{\SnRd}{\mathbb{S}^{n}(\RR^{d})}

\newcommand{\StRtr}{\mathbb{S}^{2}(\RR^{3})}

\newcommand{\VV}{\mathbb{V}} 
\newcommand{\K}{\qT{K}}
\newcommand{\J}{\qT{J}}
\newcommand{\id}{\dT{1}}
\newcommand{\I}{\mathop\mathrm{I}\limits_\approx}

\newcommand{\Ca}{\mathrm{C_{\uppercase\expandafter{\romannumeral1}}}}
\newcommand{\Cb}{\mathrm{C_{\uppercase\expandafter{\romannumeral2}}}}
\newcommand{\Cd}{\mathrm{C_{\uppercase\expandafter{\romannumeral3}}}}

\newcommand{\SO}{\mathrm{SO}}
\newcommand{\OO}{\mathrm{O}}
\newcommand{\ZZ}{\mathrm{Z}}   

\newcommand{\DD}{\mathrm{D}}

\newcommand{\oct}{\mathcal{O}}

\newcommand{\strata}[1]{\Sigma_{[#1]}}	    

\newcommand{\vg}{\mathbf{g}}

\newcommand{\G}{\mathrm{G}}

\newcommand{\Hh}{\mathrm{H}}

\DeclareMathOperator{\Orb}{Orb}      

\newcommand{\T}{\mathrm{T}}

\newcommand{\ben}{\begin{equation*}}
	\newcommand{\een}{\end{equation*}}
\newcommand{\ba}{\begin{eqnarray}}
	\newcommand{\ea}{\end{eqnarray}}
\newcommand{\ban}{\begin{eqnarray*}}
	\newcommand{\ean}{\end{eqnarray*}}

\newcommand{\red}[1]{\textcolor{red}{#1}}
\newcommand{\black}[1]{\textcolor{black}{#1}}
\newcommand{\gray}[1]{\textcolor{lightgray}{#1}}

\newcommand{\qT}[1]{\underset{\approx}{\mathrm{#1}}}
\newcommand{\tT}[1]{\underset{\simeq}{\mathrm{#1}}}
\newcommand{\dT}[1]{\underset{\sim}{\mathrm{#1}}}
\newcommand{\V}[1]{\underline{{\mathrm{#1}}}}

\newcommand{\onedot}{$\mathsurround0pt\ldotp$}
\newcommand{\dc}{
	\mathbin{\vcenter{\baselineskip.57ex
			\hbox{\onedot}\hbox{\onedot}}%
}}%

\newcommand{\qc}{
	\mathbin{\vcenter{\baselineskip.57ex
			\hbox{\onedot}\hbox{\onedot}} \vcenter{\baselineskip.57ex
			\hbox{\onedot}\hbox{\onedot}}%
}}%

\newtheorem{thm}{Theorem}[section]

\newtheorem{lem}[thm]{Lemma}

\newtheorem{prop}[thm]{Proposition}
\newtheorem{defn}[thm]{Definition}

\newtheorem{exe}[subsubsection]{Example}
\numberwithin{equation}{section}   
\newtheorem{theorem}{Theorem}[section]

\newcommand{\qact}[1]{\overline{#1}}

\newcommand{\GL}{\mathrm{GL}}
\newcommand{\vinv}{\mathbf{-1}}
\newcommand{\vr}{\mathbf{r}}

\newcommand{\HH}{\mathbb{H}} 
\newcommand{\Vv}{\mathbb{V}}

\usepackage{listings}
\usepackage{xcolor}           
\definecolor{codegreen}{rgb}{0,0.6,0}
\definecolor{codegray}{rgb}{0.75,0.75,0.75}
\definecolor{codepurple}{rgb}{0.58,0,0.82}
\definecolor{backcolour}{rgb}{0.95,0.95,0.92}
\definecolor{bubbles}{rgb}{0.91,1.0,1.0}
\definecolor{oxfordblue}{rgb}{0.0, 0.13, 0.28}
\definecolor{navyblue}{rgb}{0.0, 0.0, 0.5}
\definecolor{mintcream}{rgb}{0.96, 1.0, 0.98}
\definecolor{mediumpersianblue}{rgb}{0.0, 0.4, 0.65}
\definecolor{palecerulean}{rgb}{0.61, 0.77, 0.89}
\definecolor{midnightblue}{rgb}{0.1, 0.1, 0.44}
\definecolor{lavender(web)}{rgb}{0.9, 0.9, 0.98}
\definecolor{ForestGreen}{rgb}{0.13, 0.55, 0.13}
\definecolor{darkgreen}{rgb}{0.0, 0.5, 0.0}
\definecolor{burgundy}{rgb}{0.5, 0.0, 0.13}
\definecolor{amber}{rgb}{1.0, 0.49, 0.0}
\definecolor{brightpink}{rgb}{1.0, 0.0, 0.5}

\colorlet{red}{black}
\colorlet{blue}{black}

\begin{document}
\let\WriteBookmarks\relax
\def\floatpagepagefraction{1}
\def\textpagefraction{.001}


\author{Nicolas Auffray}
\thanks{This document exposes  results from the research
   project ANR$-19-$CE$08-0005$ (project Max-Oasis) funded by the French National Research
   Agency (ANR).}
\email{nicolas.auffray@sorbonne-universite.fr}



\title{On Exotic Materials in 3D Linear Elasticity with High Symmetry Classes}

\address{
Institut d’Alembert, CNRS UMR 7190, Sorbonne Université, Paris,75005,France}


\author[1,2]{Guangjin Mou}
\email{mouguangjin@ust.hk}
\address{Department of Civil and Environmental engineering, The Hong Kong University of Science and Technology, Hong Kong, China}


\author{Boris Desmorat}

\begin{abstract}
An anisotropic elastic material is referred to as exotic when, under specific loadings, its
mechanical response exhibits a higher degree of symmetry than that prescribed by its intrinsic
material symmetry. Such materials, which may be regarded as lying—conceptually and
functionally—between two distinct symmetry classes, are of significant practical relevance.
They enable the tailored design of metamaterials capable of reconciling otherwise incompatible
mechanical requirements; for example, achieving directional isotropy of the Young’s modulus in
an intrinsically anisotropic medium. This work focuses on the systematic classification of
exotic structures within the framework of three-dimensional linear elasticity. An exhaustive
classification is carried out, leading to the enumeration of 18 exotic structures corresponding
to symmetry classes higher than orthotropy. Representative examples of exotic elastic
behaviours are analysed in detail.
\end{abstract}



\maketitle

\section{Introduction and motivation}
\label{S_intro}

In the recent years, remarkable advancements in Additive Manufacturing (A.M.) ~\citep{wong2012}
and topology optimization \citep{bendsoe2003,Amstutz2010} have catalysed an accelerated
interest  in architectured materials. These types of materials possess special elastic
properties not found in natural materials ~\citep{schaedler2016,craster2023} such as negative
Poisson's ratio~\citep{lakes1993}, negative bulk modulus ~\citep{ding2007},
pentamode~\citep{buckmann2012}...

However, the previously listed exotic properties concern mainly isotropic elasticity. Indeed,
what makes architectured materials truly intriguing is that their unusual mechanical properties
stem not from the properties of their individual components, but rather from the intricate
internal geometry of their unit cell. Since the internal geometry is added as a new variable,
it is natural to think about the resulting anisotropy properties. Some non-standard anisotropic
linear elasticities have also been observed by researchers. In 2D, the early paper of
~\citep{Van02} has shown the existence of a particular type of planar orthotropic
material possessing an isotropic deviatoric elasticity. That is, an anisotropic orthotropic
elasticity tensor whose restriction to its purely deviatoric part is isotropic—a property that
is, generically, true only for isotropic elastic materials. The resulting tensor thus appears
to exhibit behaviour halfway between these two anisotropy classes. For historical
\textcolor{blue}{reasons}, this material has been named \emph{$R_{0}$-orthotropy}. \red{As an
example, paper exhibits an exotic elastic behavior, being $R_0$-orthotropic in compliance
\citep{HO51,OS00}.} The corresponding design has been addressed in \citep{Mou2023} and
experimental validation conducted in \cite{MDC+26}. \red{As an example, \autoref{fig:R0_example} shows
the theoretical geometry of an architectured material exhibiting $R_0$-orthotropic behavior in stiffness,
alongside the actual material obtained by 3D printing and experimentally validated in \cite{MDC+26}.}

\begin{figure}[H]
\centering
\begin{subfigure}[b]{0.32\textwidth}
\centering
\includegraphics[width=\textwidth]{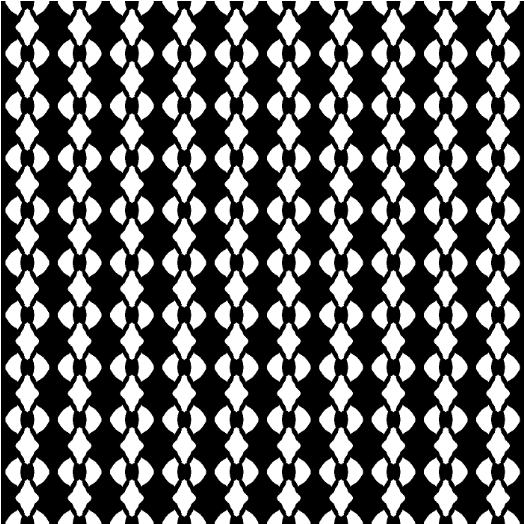}
\caption{Optimal architecture}
\label{fig:R0_theo}
\end{subfigure}
\hspace{0.04\textwidth}
\begin{subfigure}[b]{0.32\textwidth}
\centering
\includegraphics[width=\textwidth]{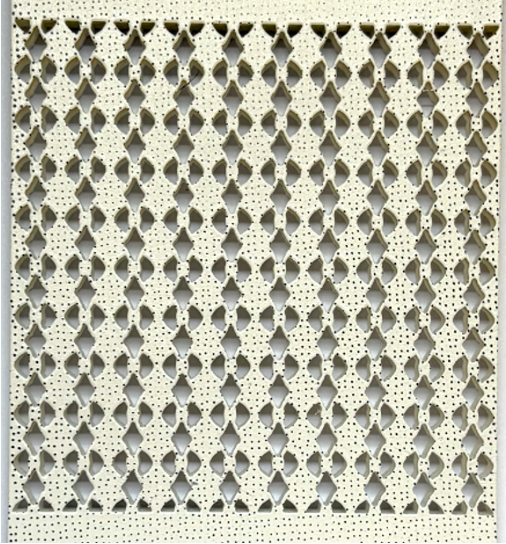}
\caption{Material obtained by 3D printing}
\label{fig:R0_fab}
\end{subfigure}
\caption{Example of a 2D $R_{0}$-orthotropic elastic material.}
\label{fig:R0_example}
\end{figure}

For the 3D case, \textcolor{blue}{Rychlewski}~\citep{Rychlewski2001} and He~\citep{He2004} have
shown, in theoretical articles, that anisotropic materials (orthotropic or transverse
isotropic) can have their directional Young's modulus, shear modulus or area modulus isotropic.
Nonetheless, these constitute nearly the only references available on the topic, and their
scope is confined to a few particular cases. They fall short of proposing a genuine theoretical
framework capable of systematically investigating this avenue, which we believe to be
particularly promising both for its theoretical richness and its potential applications.

Recent advances in A.M. technologies, coupled with developments in multi-scale topology
optimization, now offer a timely opportunity to revisit this topic and pursue further
investigations centered around the following key questions:
\begin{enumerate}
\item what is the correct mathematical definition for exotic \red{elastic} materials ?
\item can we, \emph{a priori}, list the number of exotic structures\footnote{It is equivalent to "exotic sets" defined in \citep{Mou2023}.} ?
\item what are the meso-structure producing these effects ? 
\end{enumerate}

The first question has been raised and answered in a former contribution dedicated to the
$\RR^2$ case in \citep{Mou2023}. This definition is independent of the dimension of the space, and we shall retain it
unchanged in the present contribution\footnote{It should be noted that some specifically
designed materials may have interesting non-standard properties while not meeting the
\emph{hypersymmetry} requirement. These materials will be referred to as \emph{semi-exotic}.
They will not be considered in the present publication }.
\red{
From a mathematical and technical standpoint, an exotic elastic behaviour is defined from the elasticity tensor as follows
\begin{defn}[$F$-exotic elasticity tensor]
\label{exotic_materials_math}
Let $G \leq \OO(3)$ be a closed subgroup conjugate to a representative of a
class in $\mathcal{I}(\Ela)$. We denote
\ben
\Fix(\Ela,G):=\{\qT{C}\in\Ela \mid g\star\qT{C}=\qT{C},\ \forall g\in G\}.
\een
Let $F:\Ela\longrightarrow \VV$ be an $\OO(3)$-equivariant mapping. For a
closed subgroup $H$ with $G \leq H$, we denote
\ben
\Fix(\Fix(\Ela,G),H):=\{\qT{C}\in\Fix(\Ela,G)\mid g\star F(\qT{C})=F(\qT{C}),\ \forall g\in H\}.
\een
A tensor $\qT{C}$ is said to be \emph{$G$-exotic for $F$} if there exists a
closed subgroup $H$ with $G < H$ such that
\ben
\qT{C}\in \strata{G}\cap\Fix(\Fix(\Ela,G),H)
\quad\text{and}\quad
\dim \Fix(\Fix(\Ela,G),H)<\dim\Fix(\Ela,G).
\een
\end{defn}
\noindent In this defintion $\Ela$, $\mathcal{I}(\Ela)$ and $\strata{G}$ respectively denote: the vector space of elasticity tensors, its set of symmetry classes,
 and the open stratum\\footnote{The \emph{open stratum} $\strata{G}$ denotes the set of elasticity tensors whose symmetry group is \emph{exactly} conjugate to $G$,
 as opposed to the fixed-point set $\Fix(\Ela,G)$, which also contains tensors of strictly greater symmetry.} of elasticity tensors whose symmetry group is conjugate to $G$. 
 The notation $\star$ stands for the classical tensorial action. 
These notions will be properly defined in the course of the paper. From a mechanical standpoint, this definition can be summarised more simply as follows
\begin{defn}[Exotic materials]
\label{exotic_materials3D}
An anisotropic elastic material will be said to be \emph{exotic}, provided
	\begin{enumerate}
		\item \textbf{(Specific design)} it satisfies constraints independent of those that may be imposed by symmetry arguments;
		\item \textbf{(Hypersymmetric}) behave like a material with a larger symmetry group for a restricted class of loadings.
	\end{enumerate}
\end{defn}
An orthotropic elastic material with an isotropic directional Young's modulus is indeed hyper-symmetric; in fact,
having an isotropic Young's modulus is a hallmark of isotropy.}


 \red{The resolution of the second question relies on the decomposition of the elasticity tensor into a sum of elementary tensors that 
constitute the fundamental building blocks of anisotropic behaviour. This decomposition, which lies at the very core of our approach,
is known in the literature as the harmonic decomposition \citep{Backus1970,Forte1996}. Prior to the harmonic decomposition itself is
the notion of \emph{harmonic structure}, which consists of the number and nature of the tensor spaces involved in this decomposition.
This notion is fundamental, since the classification of exotic elastic materials is based on the determination of the set of all possible exotic structures.
As a consequence of a harmonic structure involving multiple identical spaces, explicit harmonic decompositions are not uniquely determined. It follows that
the notion of exotic material must be understood relative to a particular choice of harmonic decomposition, which endows these otherwise
abstract exotic structures with a precise mechanical interpretation. We shall thus need to distinguish, in the following, between the 
notion of an \emph{exotic structure}, which identifies and labels types of degeneracy, and that of an \emph{exotic material}, which 
constitutes a concrete realisation of this structure through an explicit harmonic decomposition. Finally, exotic materials are elastic 
materials that can exhibit an \emph{exotic behaviour} under specific loadings, as stated in \autoref{exotic_materials_math}.
To sum up, a single exotic structure may lead to many different exotic materials, depending on the chosen explicit decomposition. 
In this article, among all the possibilities, we consider two distinct explicit harmonic decompositions and, on this basis, present several examples of exotic materials.}
\noindent Due to the high combinatorial complexity arising in the description of exotic structures for
low-symmetry classes, this first contribution will focus exclusively on symmetry classes of
higher symmetry than Orthotropy, namely: Trigonal $[\DD_3]$, Tetragonal $[\DD_4]$, Transversely
Isotropic (TI) $[\OO(2)]$, Cubic $[\mathcal{O}]$, and Isotropic $[\SO(3)]$. While admittedly
restrictive, this assumption still encompasses a wide range of situations of practical
relevance in engineering \citep{lemaitre1994}. The remaining symmetry classes, namely:
Orthotropic $[\DD_{2}]$, Monoclinic $[\ZZ_{2}]$, and Triclinic $[1]$ will be addressed in a
forthcoming contribution.

The last question concerns the determination of mesostructure producing the overall exotic
elasticity. It is still open and very important for practical applications. While this question
has been raised and resolved in $\RR^2$, its extension to $\RR^3$ presents inherently greater
challenges and requires a dedicated treatment, which lies beyond the scope of the present
paper. These points will be addressed in a companion paper. However, the theoretical framework
developed here provides an algebraic formulation of conditions that various exotic elastic
materials satisfy. These results are therefore essential for formulating cost functions
associated with multi-scale topology optimization problems.

However, we did not address the fundamental question of whether the proposed exotic elasticity
tensors are actually realizable. From a theoretical standpoint, this issue was investigated in
\citep{Camar-Eddine2003} in which the authors shown that any elasticity tensor can be obtained
by homogenisation. Nevertheless, this result must be nuanced in light of recent findings by
Djourachkovitch et al. who evidenced that the set of effective elasticity tensors reachable
from a two-phase material is a strict subset of $\Ela$. Such subset was qualified as
\textit{feasible elasticity tensors} \citep{djourachkovitch2023data}. At first sight, these
results may seem to be in contradiction.  However, while Djourachkovitch's findings pertain to
a specific class of meso-structures, namely, two-phase materials, Camar-Eddine's theorem is
more general, as it imposes no restrictions on the nature of the meso-structure. Consequently,
although any exotic elasticity tensor can theoretically be realized, what remains unknown in
each particular case is the nature of the required microstructure.

To summarize, the main object of the present contribution is to determine anisotropic exotic
structures in 3D linear elasticity. Yet, the analytical tools developed for this purpose bring
side results: they provide a unified perspective on the design of architectured materials
exhibiting anisotropic and exotic mechanical behaviour. Moreover, various isolated cases
reported in the literature can be recovered as particular instances of the general
constructions derived from our theoretical approach, thereby highlighting the promising scope
of future developments.

The paper is organised as follows. It starts in \autoref{S_space3D} with an overview of the
space on elasticity tensors and the linear transformations of its elements. It provides a
comprehensive introduction to the  geometrical tools required for conducting our analysis. To
ensure narrative flow, some results previously established in the literature are presented
directly, with readers encouraged to refer to the cited works for further details.  With these
tools in hand, the anisotropic exotic structures associated with tensors whose symmetry class
is higher than orthotropy are determined in \autoref{S_SymExo}. In \autoref{S_ExoMaterials}
three examples of exotic materials, based on the two explicit decomposition, are presented in
detail. Finally, the proofs of the main results are provided in \autoref{S_Proof}.


\textbf{Notations:}\\ \noindent \textit{Spaces}:
\begin{itemize}
	\item	$\mathbb{R}^{d}$: real vector space of dimension $d$;
	\item $\SnRd$ : the space of totally n-index symmetric tensors on $\mathbb{R}^{d}$;
	\item $\Ela$ : vector space of elasticity tensors on $\mathbb{R}^{3}$; 
	\item $\mathbb{H}^{n}$ : vector space of harmonic tensors of order $n$ in dimension $3$.
\end{itemize}

\noindent \textit{Operations}:
\begin{itemize}
\item	$\mathrm{tr}_{ij}\T$: trace of tensor $\T$. It is obtained by the contraction of two indices $i,j$ of $\T$;
\item $\otimes$: standard tensor product and $\otimes ^{n}$ indicates its repeated application $n$ times;
\item $\overline{\underline{\otimes}}$: tensor product  between two second order tensors 
defined by
$(\dT{a}\overline{\underline{\otimes}}\dT{b})_{ijkl}=\frac{1}{2}\left(a_{ik}b_{jl}+a_{il}b_{jk}\right)$;
\item$\oplus $ : represents the direct sum of vector spaces;
\item$\circledcirc$ : clips operation between symmetry classes.
\end{itemize}

\noindent \textit{Tensors}:
\begin{itemize}
\item Tensors of order 0,1,2,4 are respectively represented by $\alpha ,\underline{\mathrm{v}},\mathop{\mathrm{a}}\limits_\sim$, $\mathop{\mathrm{A}}\limits_\approx$;
\item$ \dT{1}$ : identity tensor of order 2, ${1}_{ij}=\delta_{ij}$ ;
\item$\I$ : identity tensor of order 4 for $\StRtr$, 
$\I= \dT{1} \, \overline{\underline{\otimes}} \, \dT{1}$;
\item$\K$ : spherical projector of $\mathbb{R}^{3}$ , $\qT{K}=\frac{1}{3}\dT{1}\otimes\dT{1}$;
\item$\J$ : deviatoric projector of $\mathbb{R}^{3}$, $\qT{J}=\qT{I}-\frac{1}{3}\dT{1}\otimes\dT{1}$;
\item $\tT{\epsilon}$ : Levi-Civita third-order tensor in $\RR^{3}$ ($\epsilon_{ijk}$ is $1$ if $(i,j,k)$ is an even permutation of $(1,2,3)$, $-1$ if it is an odd permutation, and $0$ if any index is repeated).
\end{itemize}

\section{The space of elasticity tensors}\label{S_space3D}

Let us consider $\StRtr$ the space of second-order symmetric tensors.  In the field of linear
elasticity, the constitutive law is a  linear relationship between the Cauchy stress tensor
$\dT{\sigma}\in \StRtr$ and the infinitesimal strain tensor $\dT{\varepsilon}\in \StRtr$:
	\ben\label{eq:RelEla}
		\dT{\sigma}=\qT{C}\dc\dT{\varepsilon},
	\een
in which $\qT{C}$ is the elasticity tensor, element of the following vector space:
\ben
	\Ela:=\{\qT{C}\in\otimes^4\RR^3|{C}_{\underline{(ij)}\ \underline{(kl)}}\},\quad \dim\Ela=21.
\een
Here the notation $(ij)(kl)$ indicates the minor index symmetries while $\underline{ij}\
\underline{kl}$ denotes the major one.

In order to study the anisotropic properties of elasticity tensors, it is important to
understand how the elements of $\Ela$ transform under the action of a linear isometry.

\subsection{Transformation of an elasticity tensor}\label{sub_DefSO3}

Let us define by $\OO(3)$ the group of linear isometries of $\RR^{3}$:
\ben
\OO(3):=\left\{\vg\in\GL(3) \,\middle|\, \vg^{T}=\vg^{-1}\right\},
\een
in which $\GL(3)$ is the group of invertible transformations of $\RR^{3}$.  Linear isometries
constitute the set of transformations that do not deform the geometry. Depending on the sign of
their determinant, these transformations reduce to rotations when the determinant is positive,
and to reflections, inversions, or their combinations with rotations when it is negative. For
the former use, the special orthogonal group $\SO(3)$ is introduced as the subgroup of $\OO(3)$
consisting of rotations in $\RR^3$:
\[
\SO(3) := \left\{ \vr \in \OO(3) \,\middle|\, \det \vr = 1 \right\}.
\]
The representation\footnote{A representation of a group $\G$ on a vector space $\VV$ is a group
homomorphism from $\G$ to $\GL(\VV)$.} of $\OO(3)$ on $\Ela$ is provided by the standard
tensorial action:
\ben
	\qact{\qT{C}}=\vg\star\qT{C},\ \vg\in\OO(3),
\een
in which $\qact{\qT{C}}$ is the image of $\qT{C}$ by the action of  $\vg\in\OO(3)$. This
action, when expressed in an orthonormal basis, is given in components by:
\begin{equation}
\label{Ctransformation}
    \qact{C}_{ijkl}=g_{ip}g_{jq}g_{kr}g_{ls}C_{pqrs}.
\end{equation}
The symmetry group of $\qT{C}$ is the set of operations $\vg\in\OO(3)$ leaving this tensor
invariant:
    \begin{equation}
    \label{P3C6S1_e1}
    \G_{\qT{C}}:=\{\vg\in \OO(3),\mid \qT{C}=\vg\star\qT{C}\}.
    \end{equation}
The symmetry group of $\qT{C}$ is a closed subgroup of $\OO(3)$. The classification, up to
conjugacy, of $\OO(3)$-closed subgroups is a classical result which can be found, for instance,
in \citep{sternberg1995group,golubitsky1985}.

In the case of a tensor of even order, spatial inversion ($\vg=\vinv$) is always part of its
symmetry group:
\ben
 \qact{C}_{ijkl}=(-1)^4\delta_{ip}\delta_{jq}\delta_{kr}\delta_{ls}C_{pqrs}=C_{ijkl} .
\een
This observation allows us to restrict the analysis of invariance properties to the proper
rotation group $\SO(3)$, rather than considering the full orthogonal group $\OO(3)$. Indeed,
any orthogonal transformation in $\mathbb{R}^3$ can be written in the form
\ben
\vg=\det(\vg)\ \vr(\V{n},\theta),
\een
with $\vr(\V{n}, \theta)\in\SO(3)$ a $\theta$-angle rotation around the unit vector $\V{n} \in
\RR^3$. It results that the symmetry class of $\qT{C}$ can be written as $\G_{\qT{C}} =
\mathrm{H} \otimes \mathrm{Z}^{c}_2$, in which $\mathrm{H} \in \SO(3)$ is a closed subgroup and
$\mathrm{Z}^{c}_2 = \{\mathbf{1}, \vinv \}$. Since only $\mathrm{H}$ is relevant for
classifying the symmetries, it is standard to describe the symmetry properties of even-order
tensors in terms of subgroups of $\SO(3)$\citep{Forte1996,olive2013}. In this context, the
definition in \autoref{P3C6S1_e1} can be reduced to $\SO(3)$. For simplicity sake, we will
continue to use $\vg$ to denote elements of $\SO(3)$ in what follows. Every closed subgroup of
$\SO(3)$ is conjugate to an element of the following collection \citep{armstrong2013}:
\ben
\{\left[1 \right],\left[\ZZ_n\right],\left[\DD_n\right],\left[\mathcal{T}\right],\left[\mathcal{O}\right],\left[\mathcal{I}\right],\left[\SO(2)\right],\left[\OO(2)\right],\left[\SO(3)\right]\}_{n\geq2} .
\een
This collection can be distinguished into two parts:\\ \underline{\textbf{Plane groups}:}
$\{\left[1
\right],\left[\ZZ_n\right],\left[\DD_n\right],\left[\SO(2)\right],\left[\OO(2)\right]\}_{n\geq2}$.\\
They are the closed subgroups of $\OO(2)$, with:
\begin{itemize}
    \item $1$ is the trivial subgroup, containing only the unit element;
    \item $\ZZ_{n}$ is the cyclic subgroup of $\SO(2)$ generated by $\vr(\V{n},\theta=\frac{2\pi}{n})$, with the convention that $\ZZ_{1}=1$;
    \item $\DD_{n}$ is the dihedral group generated by $\vr(\V{n},\theta=\frac{2\pi}{n})$ and $\vr(\V{k},\pi)$ ($\V{n}$ and $\V{k}$ being orthogonal);
    \item $\SO(2)$ is the group of rotations $\vr(\V{n},\theta)$, with $\theta\in[0,2\pi[$;
    \item $\OO(2)$ is the orthogonal group generated by $\vr(\V{n},\theta)$ and $\vr(\V{k},\pi)$ ($\V{n}$ and $\V{k}$ being orthogonal).
\end{itemize}
\underline{\textbf{Exceptional groups}:}
$\{\left[\mathcal{T}\right],\left[\mathcal{O}\right],\left[\mathcal{I}\right],[\SO(3)]\}$.\\
They are the groups leaving invariant the platonic solids and sphere. Specifically,
\begin{itemize}
\item $\mathcal{T}$ the tetrahedral group of order 12 which fixes a tetrahedron;
\item $\mathcal{O}$ the octahedral group of order 24 which fixes an octahedron or a cube; \item $\mathcal{I}$ the icosahedral group of order 60 which fixes the icosahedron or the dodecahedron;
\item $\SO(3)$ is the special orthogonal group that leaves the sphere invariant.
\end{itemize}

\subsection{Harmonic structure of $\Ela$}
An elasticity tensor undergoes a complex transformation that is difficult to grasp directly;
indeed, relation \autoref{Ctransformation} shows that this transformation is described by an
eighth-order tensor. To gain a better understanding of the structure of the transformation
under the action of the group $\SO(3)$, the harmonic decomposition will be introduced. This
decomposition allows the tensor space $\Ela$ to be expressed as a direct sum of harmonic
subspaces $\HH^n$, which consist of totally symmetric and traceless tensors of order $n$, and
have dimension $2n+1$.

In the following, we shall distinguish between the notion of \textit{harmonic structure} and
that of \textit{explicit harmonic decomposition}. The former refers to the decomposition of the
vector space as a direct sum of harmonic subspaces; this structure is unique. The latter, on
the other hand, consists of an explicit formula realising an isomorphism between the elasticity
tensor space and its harmonic decomposition. An explicit harmonic decomposition of $\Ela$ is
not unique and several inequivalent parametrisations are  possible.

Harmonic decompositions of $\Ela$ have been performed by many authors
\citep{Backus1970,onat1984,cowin1989,Baerheim1993,Forte1996}. All of these decompositions
correspond to the same following harmonic structure:
    \begin{equation}
    \label{harmonic_structure}
                \Ela\simeq2\mathbb{H}^{0}\oplus 2\mathbb{H}^{2}\oplus \mathbb{H}^{4},
    \end{equation}
which means that if $f$ denotes an explicit harmonic decomposition, each elasticity tensor
$\qT{C}\in \Ela$ can be written as:
	\begin{equation}
		\label{equation1}
		\qT{C}=f(\alpha,\beta,{\dT{h}}^{a},{\dT{h}}^{b},\qT{H}),
	\end{equation}
with $\alpha,\beta \in \mathbb{H}^{0} $,${\dT{h}}^{a},{\dT{h}}^{b} \in \mathbb{H}^{2} $,
$\qT{H} \in \mathbb{H}^{4}$.

An explicit harmonic decomposition is a $\SO(3)$-equivariant isomorphism meaning that
\ben
\forall \vg\in\SO(3), \quad \vg\star\qT{C}=f(\alpha,\beta,\vg\star{\dT{h}}^{a},\vg\star{\dT{h}}^{b},\vg\star\qT{H}).
\een
Based on this harmonic structure, any elasticity tensor $\qT{C}$ can be split into an
\begin{itemize}
    \item isotropic part defined by two scalars $\alpha,\beta\in \mathbb{H}^{0}$. These scalars do not transform with the tensor and hence are referred to as \textit{invariant},
    \item anisotropic part with ${\dT{h}}^{a},{\dT{h}}^{b}\in \mathbb{H}^{2}$ and $\qT{H}\in \mathbb{H}^{4}$, for which $({\dT{h}}^{a},{\dT{h}}^{b},\qT{H})$ is called a harmonic triplet in $\RR^{3}$. These tensors do transform with the tensor and hence are referred to as \textit{covariant}.
\end{itemize}
The harmonic decomposition of a tensor space carries substantial structural information; in
particular, it governs both the number and the nature of its symmetry classes\footnote{The
symmetry classes are the equivalence classes of symmetry groups with respect to rotations. This
is the correct notion associated to the intuitive one of \textit{type of anisotropy}.}
~\citep{olive2013,Olive2014c}. We will return to this point in more detail in
\autoref{Sub_SymClaEla}.

It is important to note that, since the harmonic structure (see \autoref{harmonic_structure})
contains multiple copies of harmonic spaces of the same order, the associated explicit
decomposition is not uniquely defined. The knowledge of this non-unicity is important since
many different harmonic decompositions are possible, and possess different physical meanings.
This is the same phenomenon which, in the case of isotropy, leads to multiple possible
definitions of these coefficients: $(E,\nu)$, $(\lambda,\mu)$, $(K,G)$, $\ldots$. In $\RR^3$
this freedom also applies to the parametrization of the two anisotropic\footnote{In $\RR^2$,
the harmonic structure is slightly different, and the non-uniqueness concerns only the
isotropic part \citep{Mou2023}, whereas in $\RR^3$, the situation is more complex.} subspaces
of type $\HH^2$. These different physical \textcolor{blue}{contents} will be exploited for
defining different exotic materials. In other words, it gives a mechanical interpretation of
the \emph{\textcolor{blue}{hypersymmetry}} property mentioned in the definition of exotic
materials.


\subsection{Explicit harmonic decompositions}\label{Sub_ExpDec}
Let us denote by $f$ an explicit harmonic decomposition, each elasticity tensor $\qT{C}\in
\Ela$ can be written as:
	\ben
		\qT{C}=f(\alpha,\beta,{\dT{h}}^{a},{\dT{h}}^{b},\qT{H}).
	\een
The precise definition—and thus the mechanical meaning—of $\alpha$, $\beta$, $\dT{h}^{a}$, and
$\dT{h}^{b}$ relies on the choice of an explicit harmonic decomposition. Of all admissible
decompositions, two will be retained in the sequel owing to their distinct mechanical
relevance: the Clebsch-Gordan Harmonic Decomposition (CGHD) and the Schur-Weyl Harmonic
Decomposition (SWHD).

\subsubsection*{Clebsch-Gordan Harmonic Decomposition (CGHD)}

It is based on the decomposition of the state tensor space $\StRtr$ of strain and stress into a
deviatoric space $\HH^{2}$ and a spherical one $\HH^{0}$, that is
$\StRtr\simeq\HH^{2}\oplus\HH^{0}$. Hence an elasticity tensor $\qT{C}$ is viewed as an element
of $\mathcal{L}^{s}(\HH^{2}\oplus\HH^{0})$, the space of linear symmetric applications on
$\HH^{2}\oplus\HH^{0}$. As such $\qT{C}\in\Ela$ can be decomposed into the following block
structure
\begin{equation}
\label{P3C6S3_e2}
\begin{pmatrix}
    \dT{\sigma}^{d}\\ \dT{\sigma}^{s}
\end{pmatrix}=\begin{pmatrix}
\qT{C}^{dd} &\qT{C}^{ds} \\ 
 \qT{C}^{sd} & \qT{C}^{ss}
\end{pmatrix}\begin{pmatrix}
    \dT{\varepsilon}^{d}\\ \dT{\varepsilon}^{s}
\end{pmatrix},
\end{equation}
in which $\dT{t}^{s}, \dT{t}^{d}$ denote respectively the spheric and deviatoric part of
$\dT{t}$. The elasticity tensor $\qT{C}$ decomposes into 3 different terms: $\qT{C}^{dd}$  the
deviatoric elasticity tensor, $\qT{C}^{ss}$ the spherical elasticity tensor, and $\qT{C}^{ds}$
the coupling tensor between the deviatoric and spherical parts. The remaining element is not
independent and is obtained from  $\qT{C}^{ds}$ by transposition.

The Clebsch-Gordan \textcolor{blue}{explicit} harmonic decomposition of $\Ela$ is constructed
to fit within this structure. Formally this corresponds to
\begin{equation}
\label{eq:ECGHD}
\begin{pmatrix}
    \dT{\sigma}^{d}\\ \dT{\sigma}^{s}
\end{pmatrix}=\begin{pmatrix}
\qT{H}+\dT{h}^{a}\boxtimes\dT{1}+\alpha \J &\frac{1}{3}\dT{h}^{b}\otimes\dT{1} \\ 
\frac{1}{3}\dT{1} \otimes \dT{h}^{b} & \beta\qT{K}
\end{pmatrix}\begin{pmatrix}
    \dT{\varepsilon}^{d}\\ \dT{\varepsilon}^{s}
\end{pmatrix},
\end{equation}
in which $\qT{C}^{dd}=\qT{H}+\dT{h}^{a}\boxtimes\dT{1}+\alpha\J$, $\qT{C}^{ss}=\beta\qT{K}$ and
$\qT{C}^{ds}=\frac{1}{3}\dT{h}^{b}\otimes\dT{1}$, $\qT{C}^{sd}=\frac{1}{3}\dT{1} \otimes
\dT{h}^{b}$.

$\boxtimes$ and $\overline{\underline{\otimes}}$ are tensorial products defined as follows:
\begin{align*}
\dT{a}\boxtimes\dT{b}&=\frac{1}{7}\left(
6\left(\dT{a}\ \overline{\underline{\otimes}}\ \dT{b}+\dT{b}\ \overline{\underline{\otimes}}\ \dT{a}\right)
-4\left(\dT{a}\otimes\dT{b}+\dT{b}\otimes\dT{a}\right)
\right),
&
\left(\dT{a}\,\overline{\underline{\otimes}}\,\dT{b}\right)_{ijkl}&=\frac{1}{2}\left(a_{ik}b_{jl}+a_{il}b_{jk}\right).
\end{align*}

The explicit Clebsch-Gordan harmonic decomposition is the object of the following proposition:
\begin{prop}
 \label{P3C6_propo1}
The tensor $\qT{C}\in\Ela$ admits the uniquely defined Clebsch-Gordan Harmonic Decomposition
associated with the use of deviatoric projector $\qT{J}=\qT{I}-\frac{1}{3}\dT{1}\otimes\dT{1}$
and spherical projector $\qT{K}=\frac{1}{3}\dT{1}\otimes\dT{1}$:
  
		\begin{equation}
			\label{P3C6S3_e5}
			\qT{C}=\alpha\J +\beta \K+\dT{h}^{a}\boxtimes\dT{1}+\frac{1}{3}(\dT{h}^{b}\otimes\dT{1}+\dT{1}\otimes\dT{h}^{b})+\qT{H},
		\end{equation}
in which  $\left(\alpha, \beta ,\dT{h}^{a},\dT{h}^{b},\qT{H}\right)\in \HH^0 \times \HH^0\times
\HH^2\times \HH^2 \times \HH^4$.

Conversely, the different elements of the decomposition can be computed as a function of
$\qT{C}$:
\begin{center}
			\begin{tabular}{|c|c|c|}
				\hline
				$\HH^{0}$ & $\HH^{2}$ & $\HH^{4}$  \\
				\hline
                $\alpha=\frac{1}{5}\qT{C}^{dd}\qc\J$&  $\dT{h}^{a}=(\tr_{13}\qT{C}^{dd}):\qT{J}$  &$\qT{H}=\qT{C}^{dd}-\dT{h}^{a}\boxtimes\dT{1}-\alpha \J$\\  
				$\beta=\frac{1}{3}\dT{1}:\qT{C}:\dT{1}$  &$\dT{h}^{b}=\J:\qT{C}:\id$ & \\       
				\hline
			\end{tabular}
		\end{center}
in which $\qT{C}^{dd}=\J:\qT{C}:\J$ is the deviatoric elasticity tensor.
\end{prop}

From a mechanical standpoint, this construction generalises the classical Bulk–Shear
parametrisation of isotropic tensors~\citep{Auffray2017b}. In the case of isotropy, the
Clebsch-Gordan decomposition reduces to:
\ben
\qT{C}=\alpha\qT{J}+\beta\qT{K},
\een
which corresponds to the \emph{shear modulus} $(G)$ and the \emph{bulk modulus} $(K)$ with
$\alpha=2G$ and $\beta=3K$.

\subsubsection*{Schur-Weyl Harmonic Decomposition (SWHD)}
This decomposition is based on the decomposition of the elasticity tensor into a totally
symmetric (index symmetry) part $\qT{C}^{s}\in {\mathbb{S}}^{4}(\RR^{3})$ and an asymmetric one
$\qT{C}^{a}\in \mathbb{S}^{2}({\mathbb{S}}^{2}(\RR^{3}))$ \red{has been introduced by} \cite{Backus1970}. The corresponding
spaces are orthogonal, we have:
\ben
\qT{C}=\qT{C}^{s}+\qT{C}^{a},
\een
with $C^{s}_{ijkl}=\frac{1}{3}(C_{ijkl}+C_{ikjl}+C_{iljk})$, and
$C^{a}_{ijkl}=\frac{1}{3}(2C_{ijkl}-C_{ikjl}-C_{iljk})$. The tensors $\qT{C}^{s}$ and
$\qT{C}^{a}$ are reducible and decompose into a direct sum of harmonic components as follows:
\begin{align*}
\qT{C}^{s}&=\beta\dT{1}\otimes_{(4)}\dT{1}+\dT{h}^{b}\otimes_{(4)}\dT{1}+\qT{H},
&
\qT{C}^{a}&=\alpha\dT{1}\otimes_{(2,2)}\dT{1}+\dT{h}^{a}\otimes_{(2,2)}\dT{1}.
\end{align*}
The symmetrized tensor products $\otimes_{(4)}$ and $\otimes_{(2,2)}$ between two symmetric
second-order tensors $(\dT{t_{1}},\dT{t_{2}})$ are defined as:
        \begin{align*}
\dT{t_{1}}\otimes_{(4)}\dT{t_{2}}&=\frac{1}{6}\left(\dT{t_{1}}\otimes\dT{t_{2}}+\dT{t_{2}}\otimes\dT{t_{1}}+2\dT{t_{1}}\ \overline{\underline{\otimes}}\ \dT{t_{2}}+2\dT{t_{2}}\ \overline{\underline{\otimes}}\ \dT{t_{1}}\right),
\\
\dT{t_{1}}\otimes_{(2,2)}\dT{t_{2}}&=\frac{1}{3}\left(\dT{t_{1}}\otimes\dT{t_{2}}+\dT{t_{2}}\otimes\dT{t_{1}}-\dT{t_{1}}\ \overline{\underline{\otimes}}\ \dT{t_{2}}-\dT{t_{2}}\ \overline{\underline{\otimes}}\ \dT{t_{1}}\right).
        \end{align*}
        
The explicit Schur-Weyl harmonic decomposition is the object of the following proposition:
    \begin{prop}
    \label{prop2}
        The tensor $\qT{C}\in\Ela$ admits the uniquely defined Schur-Weyl Harmonic
        Decomposition:
        \begin{equation}
        \label{P3C6S3_e10}
\qT{C}=\alpha\dT{1}\otimes_{(2,2)}\dT{1}+\beta\dT{1}\otimes_{(4)}\dT{1}+\dT{h}^{a}\otimes_{(2,2)}\dT{1}+\dT{h}^{b}\otimes_{(4)}\dT{1}+\qT{H},
        \end{equation}
        in which $\left(\alpha, \beta ,\dT{h}^{a},\dT{h}^{b},\qT{H}\right)\in \HH^0 \times
        \HH^0\times \HH^2\times \HH^2 \times \HH^4$. Conversely, the different elements of the
        decomposition can be computed as a function of $\qT{C}$:
\begin{center}
\begin{tabular}{|c|c|c|}
\hline
$\HH^{0}$ & $\HH^{2}$ & $\HH^{4}$  \\
\hline
$\alpha=\frac{1}{4}\tr^{(2)}\qT{C}^{a}$&  
$\dT{h}^{a}=3\left(\tr_{12}\qT{C}^{a}-\frac{4}{3}\alpha\dT{1}\right)$  &\\                    
$\beta=\frac{1}{5}\tr^{(2)}\qT{C}^{s}$  &
$\dT{h}^{b}=\frac{6}{7}\left(\tr_{12}\qT{C}^{s}-\frac{5}{3}\beta\dT{1}\right)$ &
$\qT{H}=\qT{C}^{s}-\dT{h}^{b} \otimes_{(4)} \dT{1}-\beta\dT{1} \otimes_{(4)} \dT{1}$  \\ 
\hline
\end{tabular}
\end{center}
where $\tr^{(2)}$ represents a double trace, we have $\tr^{(2)}\qT{C}^{s}=C^{s}_{iijj}$.
\end{prop}

This decomposition has a physical meaning with regard to the propagation of
waves~\citep{itin2013,itin2015}, damage mechanics~\citep{desmorat2016,ken1984,onat1984} and
\textcolor{blue}{is} also used to define an exotic anisotropic material with isotropic Young's
modulus \citep{Ryc01, He2004}.

From a historical perspective, the $\qT{C}^{s}$ component corresponds to the elasticity model
developed by Navier and Cauchy, based on atomistic considerations with central-force
potentials. Owing to the presence of a single isotropic elastic constant, this model is
sometimes referred to in the literature as the rari-constant model. It stood in contrast to an
alternative elasticity framework derived from thermodynamic principles, notably advanced by
Green. From the viewpoint of lattice dynamics, the $\qT{C}^{a}$ component quantifies the extent
to which the underlying “atomistic potential” deviates from a purely central-force interaction.
Further details on this historical controversy can be found in the following references
\citep{poincare1892,love1944,capecchi2011}.

\section{From symmetry classes to exotic structures}
\label{S_SymExo}

As will become apparent, identifying exotic structures is intimately linked to the study of the
symmetry classes of $\Ela$. We shall thus begin by addressing this aspect before presenting our
main results on exotic structures.

\subsection{Definition of symmetry classes}
\label{Sub_DefSymCla}
Let $\qT{C}$ be an elasticity tensor, its   orbit with respect to $\SO(3)$, denoted by
$\Orb(\qT{C})$, is the set of all images of $\qT{C}$ under the elements of $\SO(3)$
\citep{Abud1983,olive2017a}:
	\begin{equation}
	    	\Orb(\qT{C})=\left \{ \qact{\qT{C}}\in \Ela\mid \qact{\qT{C}}=\vg\star \qT{C} ,  \vg\in\SO(3) \right \}.
	\end{equation}
From a physical point of view, this set represents all elasticity tensors corresponding to the
same elastic material.  Alternatively, one may interpret the orbit of $\qT{C}$ as encoding the
elastic properties of a material when its orientation in space has been "forgotten" or
disregarded.

Symmetry groups of elasticity tensors along an orbit are of the same kind, distinguished only
by their orientation. More precisely, these symmetry groups are conjugate to one another,
meaning that they satisfy the relation:
\ben
\forall \qT{\qact{C}}\in\Orb(\qT{C}),\ 
\exists \vg\in \SO(3) \mid  \mathrm{G}_{ \qT{\qact{C}}}=\vg\mathrm{G}_{\qT{C}}\vg^{-1}.
\een
A weaker equivalence relation than being on the same orbit consists only in having a conjugate
symmetry group:
\ben
 \qT{\qact{C}}\approx \qT{C} \quad \Leftrightarrow \quad\{ \exists \vg\in \SO(3) \mid  \mathrm{G}_{ \qT{\qact{C}}}=\vg\mathrm{G}_{\qT{C}}\vg^{-1}\}.
\een
This relation indicates that two tensors are equivalent if their symmetry groups are conjugate.
The symmetry class $[\mathrm{G}_{\qT{C}}]$ is the conjugacy classes of $\mathrm{G}_{\qT{C}}$:
	\ben
	[\mathrm{G}_{\qT{C}}]=\{\vg \mathrm{G}_{\qT{C}} \vg^{-1}, \mid \vg\in \SO(3) \}.
	\een




\subsection{Symmetry classes of $\Ela$}
\label{Sub_SymClaEla}

The symmetry classes of \( \Ela \) were first determined by \citet{Forte1996}. They proved the
following result:
\begin{theorem}\label{th:ClaEla}
With respect to $\SO(3)$, the symmetry group of any elasticity tensor of $\Ela$ belongs to one
of the following 8 symmetry classes:
\ben
\{[1],[\ZZ_2],[\DD_2],[\DD_3],[\DD_4],[\OO(2)],[\mathcal{O}],[\SO(3)]\}.
\een
\end{theorem}
The relationships between the different symmetry classes are represented in
\autoref{Fig_poset}, on which an arrow from class $[\mathrm{H}_1]$ to class $[\mathrm{H}_2]$
indicates that $\mathrm{H}_1$ is conjugate to a subgroup of $\mathrm{H}_2$.
\begin{figure}[tbp]
\centering
\includegraphics[width=0.45\textwidth]{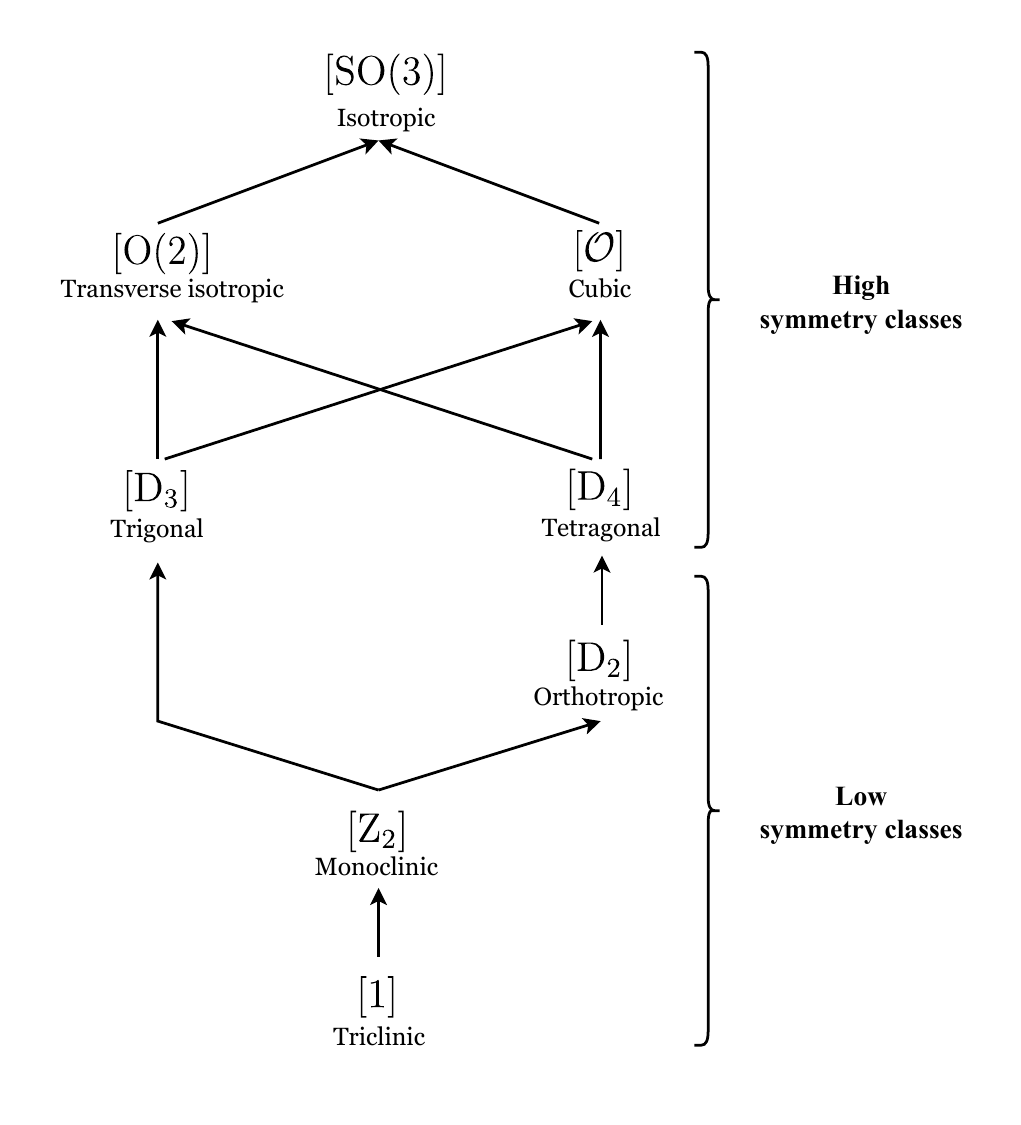}
\caption{The structure of symmetry classes of $\Ela$ as a partially ordered set (poset).}
\label{Fig_poset}
\end{figure}
In this poset, we define symmetry classes above orthotropy $[\DD_2]$ as \emph{high}, with
\emph{low symmetry classes} being triclinic $[1]$, monoclinic $[\ZZ_2]$, and orthotropic
$[\DD_2]$.

In \citet{olive2013, olive2014bs}, the authors proposed a general method for computing the
symmetry classes from the harmonic structure of any tensor space. Central to this method is the
use of the \emph{clips operation}, which allows one to determine the symmetry classes of a
direct sum $\Vv_{1} \oplus \Vv_{2}$ from the symmetry classes of vector spaces $\Vv_{1}$ and
$\Vv_{2}$. To maintain continuity,  technical definitions and properties of clips product are
provided in the \autoref{App_cliOpe}. This approach is effective since the sets of symmetry
classes of $\HH^{n}$ has been determined by ~\cite{ihrig1984}. Denoting $\SymC(\HH^{n})$ the
set of symmetry classes of $\HH^{n}$, the results for $\HH^{2}$ and $\HH^{4}$ are:
\begin{equation}
\label{eq_symmetryclasshs}
    \begin{split}
\SymC(\HH^{2})&=\{[\DD_2],[\OO(2)],[\SO(3)]\} \\
\SymC(\HH^{4})&=\{[1],[\ZZ_2],[\DD_2],[\DD_3],[\DD_4],[\OO(2)],[\mathcal{O}],[\SO(3)]\}
\end{split}
\end{equation}

The result  of \autoref{th:ClaEla} is  directly retrieved by applying clips operations detailed
in \autoref{clipspro}.

\begin{theorem}[Clips operations between $\SO(3)$-closed subgroups,~\cite{Olive2019}]
\label{clipspro}
\small
\ben
\begin{tabular}{|c|c|c|c|c|c|c|c|}
\hline
$\circledcirc$&$[\ZZ_n]$& $[\DD_n]$&$[\mathcal{T}]$&$[\mathcal{O}]$&$[\mathcal{I}]$&$[\SO(2)]$&$[\OO(2)]$ 
\\ \hline
$[\ZZ_m]$& $[1],[\ZZ_d]$& &\multirow{2}{*}{}&\multirow{3}{*}{} & \multirow{4}{*}{}&\multirow{5}{*}{} &\multirow{6}{*}{}
\\ \cline{1-3}
$[\DD_m]$ &\makecell[c]{$[1],[\ZZ_d]$\\ $[\ZZ_{d_{2}}]$}  &\makecell[c]{$[1]$\\$[\ZZ_2],[\ZZ_d]$\\$[\DD_{dz}],[\DD_d]$} & &  & 
 &&
 \\ \cline{1-4}
$[\mathcal{T}]$ & \makecell[c]{$[1]$\\ $[\ZZ_{d_{2}}]$ \\ $[\ZZ_{d_{3}}]$}& \makecell[c]{$[1],[\ZZ_2]$ \\ $[\ZZ_{d_{3}}],[\DD_{d_{2}}]$} & \makecell[c]{$[1],[\ZZ_2]$\\ $[\ZZ_{3}],[\DD_{2}]$\\$[\mathcal{T}]$} &&&&  
\\  \cline{1-5}
$[\mathcal{O}]$&\makecell[c]{$[1]$\\ $[\ZZ_{d_{2}}]$ \\ $[\ZZ_{d_{3}}]$\\$[\ZZ_{d_{4}}]$} & \makecell[c]{$[1],[\ZZ_2]$ \\ $[\ZZ_{d_{3}}],[\ZZ_{d_{4}}]$\\$[\DD_{d_{2}}],[\DD_{d_{3}}]$\\$[\DD_{d_{4}}]$} & \makecell[c]{$[1],[\ZZ_2]$\\ $[\ZZ_{3}],[\DD_{2}]$\\$[\mathcal{T}]$} & \makecell[c]{$[1],[\ZZ_2]$\\ $[\ZZ_{3}],[\ZZ_{4}]$\\$[\DD_{2}],[\DD_{3}]$\\$[\DD_{4}],[\mathcal{O}]$}&&&
\\ \cline{1-6}
$[\mathcal{I}]$&\makecell[c]{$[1]$\\ $[\ZZ_{d_{2}}]$ \\ $[\ZZ_{d_{3}}]$\\$[\ZZ_{d_{5}}]$} & \makecell[c]{$[1],[\ZZ_2]$ \\ $[\ZZ_{d_{3}}],[\ZZ_{d_{5}}]$\\$[\DD_{d_{2}}],[\DD_{d_{3}}]$\\$[\DD_{d_{5}}]$} & \makecell[c]{$[1],[\ZZ_2]$\\ $[\ZZ_{3}],[\mathcal{T}]$} &\makecell[c]{$[1],[\ZZ_2]$\\ $[\ZZ_{3}],[\DD_{3}]$\\ $[\mathcal{T}]$} & \makecell[c]{$[1],[\ZZ_2]$\\ $[\ZZ_{3}],[\ZZ_{5}]$\\$[\DD_{3}],[\DD_{5}]$\\$[\mathcal{I}]$}&&
\\ \cline{1-7}
$[\SO(2)]$ & $[1],[\ZZ_n]$ & \makecell[c]{$[1],[\ZZ_2]$\\$[\ZZ_n]$} & \makecell[c]{$[1],[\ZZ_2]$\\$[\ZZ_3]$} &\makecell[c]{$[1],[\ZZ_2]$\\$[\ZZ_3],[\ZZ_4]$}&\makecell[c]{$[1],[\ZZ_2]$\\$[\ZZ_3],[\ZZ_5]$}  &\makecell[c]{$[1]$\\$[\SO(2)]$} &
\\ \cline{1-8}
$[\OO(2)]$ &\makecell[c]{$[1],[\ZZ_{d_{2}}]$\\$[\ZZ_n]$} & \makecell[c]{$[1],[\ZZ_2]$\\$[\DD_{k_{2}}],[\DD_n]$}& \makecell[c]{$[1],[\ZZ_2]$\\$[\ZZ_3],[\DD_{2}]$} &\makecell[c]{$[1],[\ZZ_2]$\\$[\DD_2],[\DD_3]$\\$[\DD_{4}]$} & \makecell[c]{$[1],[\ZZ_2]$\\$[\DD_2],[\DD_3]$\\$[\DD_{5}]$} &\makecell[c]{$[1],[\ZZ_2]$\\$[\SO(2)]$} &\makecell[c]{$[\ZZ_2],[\DD_2]$\\$[\OO(2)]$}
\\ \hline
\end{tabular}
\een
with the following notations:
\begin{align*}
d&:=\gcd(m,n),    \quad &d_{2}&:=\gcd(n,2), \quad  &k_{2}&:=3-d_{2},\\
d_{3}&:=\gcd(n,3),\quad &d_{5}&:=\gcd(n,5), &\ZZ_{1}&=1, \\
dz&:=2, \ \textrm{if $m$ and $n$ even}, \quad  &dz&:=1, \ \textrm{otherwise,} \quad&& \\
d_4&:=4, \ \textrm{if $4$ divide $n$},  \quad  &d_4&:=1, \ \textrm{otherwise.}\quad &&
\end{align*}
in which $gcd$ stands for greater common divisor.
\end{theorem}


Beyond simply stating that there are 8 different symmetry classes, the clips operation provides
more refined information about the harmonic structures associated with each one of these
classes. Indeed, one can observe that each symmetry classes can be obtained from various
combinations of symmetry classes of the covariant within the harmonic triplet
$(\dT{h}^{a},\dT{h}^{b},\qT{H})$. To be more specific,  the symmetry class for $\qT{C}$ is the
result of:
\begin{enumerate}
    \item the symmetry classes of the individual covariants:
    \ben
    \{[\G_{\dT{h}^{a}}],[\G_{\dT{h}^{b}}],[\G_{\qT{H}}]\},
    \een
    \item the symmetry classes of covariants taken pairwise:
    \ben
\{[\G_{(\dT{h}^{a},\dT{h}^{b})}],[\G_{(\dT{h}^{a},\qT{H})}],[\G_{(\dT{h}^{b},\qT{H})}]\},
\een
with
\ben
 [\G_{(\dT{h}^{a},\dT{h}^{b})}]\in[\G_{\dT{h}^{a}}]\circledcirc [\G_{\dT{h}^{b}}], \quad [\G_{(\dT{h}^{a},\qT{H})}]\in[\G_{\dT{h}^{a}}]\circledcirc [\G_{\qT{H}}], \quad [\G_{(\dT{h}^{b},\qT{H})}]\in[\G_{\dT{h}^{b}}]\circledcirc [\G_{\qT{H}}] .
\een
\end{enumerate}
As a result, to each elasticity tensor $\qT{C}$, we associate the following ordered list:
\begin{equation}
\label{P3C7S1e1}
\{\underbrace{[\G_{\dT{h}^{a}}],[\G_{\dT{h}^{b}}],[\G_{\qT{H}}]}_{singletons},\underbrace{[\G_{(\dT{h}^{a},\dT{h}^{b})}],[\G_{(\dT{h}^{a},\qT{H})}],[\G_{(\dT{h}^{b},\qT{H})}]}_{pairs},\underbrace{[\G_{( \dT{h}^{a},\dT{h}^{b},\qT{H})}]}_{triplet}\}.
\end{equation}
This list, hereafter referred to as the \emph{geometric structure}, constitutes the geometric
identity card of an elasticity tensor. This geometric structure can also be represented (see
\autoref{fig:Simplex}) as a chain of simplices: the 2-simplex encodes the class of the triplet,
its boundary consists of three 1-simplices representing the classes of the pairs, and the
boundary of the boundary consists of points representing the classes of the elementary
covariants.


Among the various geometric structures, and for each of the 8 symmetry classes, we can identify
one structure that we shall refer to as \emph{generic}.
\begin{defn}[Generic structure and exotic structures]
\label{def_generac_exotic}
For a prescribed symmetry class $[\G]$ with $\G \in \SO(3)$, there exists:
	\begin{enumerate}
    \item \textbf{one generic structure:} the geometric structure realizing the symmetry class $[\G]$ for which all components and pairs exhibit minimal symmetry.  A randomly chosen elasticity tensor belonging to the symmetry class $[\G]$ is almost surely generic.
    \item \textbf{exotic structures:} all other structures in which at least one element is more symmetric than strictly required for belonging to the symmetry class $[\G]$.
	\end{enumerate}
\end{defn}
For instance, a generic triclinic tensor refers to an elasticity tensor obtained by randomly selecting an
element from $\Ela$. We define below its \textit{generic structure}:
\begin{enumerate}
    \item each covariant should be of least symmetry:
    \ben
    \{[\G_{\dT{h}^{a}}],[\G_{\dT{h}^{b}}],[\G_{\qT{H}}]\}=\{[\DD_2],[\DD_2], [1]\},
\een
    \item the covariants should be in arbitrary relative orientations with respect to one another:
    \ben
    \{[\G_{(\dT{h}^{a},\dT{h}^{b})}],[\G_{(\dT{h}^{a},\qT{H})}],[\G_{(\dT{h}^{b},\qT{H})}]\}=\{[1],[1],[1]\}.
\een
\end{enumerate}
The geometric structure is illustrated in \autoref{fig:SimplexGeneric}.

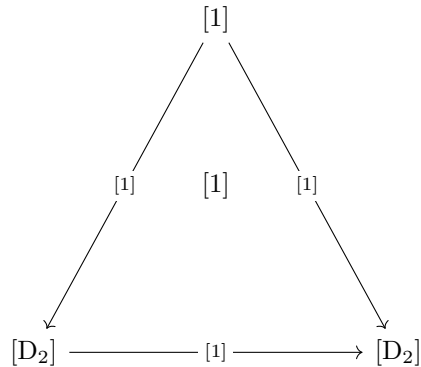
\begin{figure}[tbp]
\centering
\begin{subfigure}[b]{0.49\textwidth}
\centering
\resizebox{!}{101pt}{ \begin{tikzcd} [sep =  0.8 cm]
 & & \black{\left[\mathrm{G}_{\qT{H}}\right]}
 \arrow[ddddll, ->, black,"{[\G_{(\dT{h}^{a},\qT{H})}]}" description]
 \arrow[ddddrr, ->, black, "{[\G_{(\dT{h}^{b},\qT{H})}]}" description]
 & &   \\
&&&& \\
 & & \black{\left[\G_{(\dT{h}^{a},\dT{h}^{b},\qT{H})}\right]}  & &   \\
 &&&& \\
\black{\left[\mathrm{G}_{\dT{h}^{a}}\right]}
\arrow[rrrr, ->, black,"{[\G_{(\dT{h}^{a},\dT{h}^{b})}]}" description]
& & & & \black{\left[\mathrm{G}_{\dT{h}^{b}}\right]}  \\
 \end{tikzcd}
 }
\caption{The \emph{geometric structure} as a sequence of simplices.}
\label{fig:Simplex}
\end{subfigure}
\hfill
\begin{subfigure}[b]{0.49\textwidth}
\centering
\resizebox{!}{101pt}{ \begin{tikzcd} [sep =  0.8 cm]
 & & \black{\left[1\right]} \arrow[ddddll, ->,black,"{[1]}" description] \arrow[ddddrr, ->,black, "{[1]}" description] & &   \\
&&&& \\
 & & \black{\left[1\right]}  & &   \\
 &&&& \\
\black{\left[\DD_2\right]} \arrow[rrrr,->, black,"{[1]}" description] & & & & \black{\left[\DD_2\right]}  \\
 \end{tikzcd}
 }
\caption{The \emph{generic structure} of a triclinic elasticity tensor.}
\label{fig:SimplexGeneric}
\end{subfigure}
\caption{Geometric structure of an elasticity tensor represented as a sequence of simplices (a), and the generic structure of a triclinic elasticity tensor (b).}
\label{fig:Simplex_combined}
\end{figure}

Any other geometric structure leading to a triclinic elasticity tensor will be regarded as
\textit{non-generic}, and will therefore be referred to as \textit{exotic} in the present
context.

Based on this definition, we can determine the generic anisotropic structures corresponding to
each of the 7 non-trivial symmetry classes. Their geometric structures are obtained by
symmetrizing a generic triclinic tensor with respect to a representative of the group $\mathrm
G$.

\begin{prop}\label{prop1}
The generic geometric structure of the 8 symmetry classes of $\Ela$ are provided in the table
below.
\begin{center}
    
  \begin{tabular}{|c||c|c|c|c|c|c|}
 		\hline		
$[\G]$&$[\G_{\dT{h}^{a}}]$ &$[\G_{\dT{h}^{b}}]$ &$[\G_{\qT{H}}]$&$[\G_{(\dT{h}^{a},\dT{h}^{b})}]$&$[\G_{(\dT{h}^{a},\qT{H})}]$&$[\G_{(\dT{h}^{b},\qT{H})}]$ \\
 		\hline
            \hline 

$[\SO(3)]$&$[\SO(3)]$&$[\SO(3)]$&$[\SO(3)]$&$[\SO(3)]$&$[\SO(3)]$&$[\SO(3)]$ \\
 		\hline

$[\mathcal{O}]$&$[\SO(3)]$&$[\SO(3)]$&$[\mathcal{O}]$&$[\SO(3)]$&$[\mathcal{O}]$&$[\mathcal{O}]$  \\
\hline

 $[\OO(2)]$&$[\OO(2)]$&$[\OO(2)]$ &$[\OO(2)]$&$[\OO(2)]$&$[\OO(2)]$&$[\OO(2)]$\\
 \hline

$[\DD_4]$& $[\OO(2)]$ &$[\OO(2)]$ &$[\DD_4]$&$[\OO(2)]$&$[\DD_4]$&$[\DD_4]$ \\
  \hline
  
$[\DD_3]$& $[\OO(2)]$
&$[\OO(2)]$&$[\DD_3]$&$[\OO(2)]$&$[\DD_3]$&$[\DD_3]$\\	
 			\hline
            
$[\DD_2]$& $[\DD_2]$
&$[\DD_2]$&$[\DD_2]$&$[\DD_2]$&$[\DD_2]$&$[\DD_2]$\\	
\hline

$[\ZZ_2]$& $[\DD_2]$
&$[\DD_2]$&$[\ZZ_2]$&$[\ZZ_2]$&$[\ZZ_2]$&$[\ZZ_2]$\\
\hline

${[1]}$& $[\DD_2]$
&$[\DD_2]$&$[1]$&$[1]$&$[1]$&$[1]$\\
\hline
\end{tabular}
\end{center}
\end{prop}
\noindent The proof is given in \autoref{S_Proof}.

\subsection{Exotic structures for symmetry classes higher than orthotropic}
\label{Sub_ExoSet}
The purpose of this subsection is to highlight the two elementary mechanisms by which an exotic
structure can emerge from a generic one\footnote{In $\mathbb{R}^2$, as studied in
\citep{Mou2023}, the degeneracy of elasticity tensors arises from two distinct mechanisms: (1)
\emph{cancellation}, and (2) \emph{alignment}. As detailed below, the degeneracies observed in
$\mathbb{R}^3$ naturally extend this framework.}. It should be noted that certain combinations
of these degeneracies may lead to a change in the symmetry class of the elasticity tensor. The
two fundamental mechanisms are as follows\textcolor{blue}{:}
\begin{itemize}
    \item \textbf{Type I (Norms): Degeneracy on individual components $\{\dT{h}^{a},\dT{h}^{b},\qT{H}\}$}:\\
For vectors, i.e., elements of $\HH^1$, only two symmetry classes are possible: isotropic
$[\SO(3)]$ and anisotropic $[\OO(2)]$. A vector is either non-zero, in which case its symmetry
class is $[\OO(2)]$, or it vanishes identically, resulting in the fully isotropic class
$[\SO(3)]$. In contrast, for elements in $\HH^n$ with $n > 1$, the situation becomes more
nuanced, as partial cancellations may occur giving rise to a more intricate hierarchy of
symmetry classes. For example, in the case of a second-order harmonic tensor $\dT{h} \in
\mathbb{H}^2$, one observes the chain of symmetry classes:
\ben
[\DD_{2}]\longrightarrow[\OO(2)]\longrightarrow[\SO(3)]. 
\een
The lattice of partial cancellations becomes even more elaborate for fourth-order tensors
$\mathbb{H}^4$, a situation thoroughly examined in \cite{auffray2014}. The presence of these
\textcolor{blue}{degeneracies} is revealed by the first three entries of the geometric
structure list.



\item \textbf{Type II (Relative Orientation): Degeneracy on coupled pairs $\{(\dT{h}^{a},\dT{h}^{b}),(\dT{h}^{a},\qT{H}),(\dT{h}^{b},\qT{H})\}$}:\\
This situation accounts for the relative orientation between covariants. For vectors, i.e.,
elements of $\HH^1$, the situation \textcolor{blue}{is} relatively simple: they are either
aligned or not, and their relative alignment can be described by angles. For $\HH^n$, $n>1$ the
situation is, as expected, more complex, since elements in $\HH^2$ are five-dimensional, and
those in $\HH^4$ are nine-dimensional.

To illustrate this, let us consider the case of a pair of elements $(\dT{h}^{a}, \dT{h}^{b})$
in $\HH^2$. For the sake of the example, these covariants will be generic, that is
$[\G_{\dT{h}^{a}}] = [\G_{\dT{h}^{b}}] = [\DD_2]$. Since elements in $\HH^2$ are symmetric
second-order tensors, their relative orientation can be understood through their eigenvectors.
Accordingly, the symmetry class of the pair, indicated by $[\G_{(\dT{h}^{a}, \dT{h}^{b})}]$ is
either:
\begin{itemize}
\item $[\DD_2]$ if $\dT{h}^{a}$ and $\dT{h}^{b}$ share the same set of eigenvectors;
\item $[\ZZ_2]$ if $\dT{h}^{a}$ and $\dT{h}^{b}$ share exactly one eigenvector;
\item $[1]$ if $\dT{h}^{a}$ and $\dT{h}^{b}$ have no eigenvectors in common.
\end{itemize}
The situation becomes significantly more intricate when the covariant $\qT{H} \in \HH^4$ is
taken into account.
\end{itemize}

Thus, for a given symmetry class of the triplet $(\dT{h}^{a},\dT{h}^{b}, \qT{H})$, we will
refer to as degenerate any type I or type II modification of the generic structure that does
not result in a change of the initial symmetry class. The degeneracy mechanisms can be directly
read off from the \emph{geometric structure}:
\begin{equation}
\{\underbrace{[\G_{\dT{h}^{a}}],[\G_{\dT{h}^{b}}],[\G_{\qT{H}}]}_{Type I},\underbrace{[\G_{(\dT{h}^{a},\dT{h}^{b})}],[\G_{(\dT{h}^{a},\qT{H})}],[\G_{(\dT{h}^{b},\qT{H})}]}_{Type II},\underbrace{[\G_{(\dT{h}^{a},\dT{h}^{b},\qT{H})}]}_{Result}\}.
\end{equation}
These two types of degeneracies give rise to multiple \textit{geometric structures}, all
yielding the same overall symmetry class $[\G]$. This forms the fundamental basis of what we
refer to as \emph{exotic structures}.

This naturally leads to the question of how many different geometric structures can be
identified. To answer this question, it should first be noted that not every “geometric
structure” is admissible. Indeed, while the first three entries of such a structure—related to
the norms of the covariants—may be freely assigned, the entries describing their relative
orientations are not. These are constrained:
\begin{itemize}
\item by the symmetry classes of the covariants. For example, if two covariants are individually invariant under $\SO(3)$,
their combined pair cannot exhibit any symmetry other than $\SO(3)$. These compatibility constraints are explicitly given by the clip tables;
\item by pairwise compatibility: once the relative orientation of two pairs is fixed, the third one is necessarily constrained.
To be more specific, \textcolor{red}{choosing}  $[\G_{(\dT{h}^{a},\dT{h}^{b})}]\in[\G_{\dT{h}^{a}}]\circledcirc [\G_{\dT{h}^{b}}], \quad [\G_{(\dT{h}^{a},\qT{H})}]\in[\G_{\dT{h}^{a}}]\circledcirc [\G_{\qT{H}}]$ may restrict admissible choices for $[\G_{(\dT{h}^{b},\qT{H})}]=[\G_{\dT{h}^{b}}]\circledcirc [\G_{\qT{H}}]$, as the symmetry operations of $\dT{h}^{b}$ and $\qT{H}$ must align with those already defined by $[\G_{(\dT{h}^{a},\dT{h}^{b})}]$ and $[\G_{(\dT{h}^{a},\qT{H})}]$. This constraint is not captured by the clip tables and must be enforced explicitly, particularly for the pairs used to generate low symmetry classes.
\end{itemize}

\noindent As a consequence, the complete enumeration of exotic structures cannot be carried out
by a purely combinatorial approach\footnote{The situation is different from the $\mathbb{R}^2$
case for which direct combinatorial enumeration was possible.}. But, for the symmetry classes higher
than $[\DD_2]$, the problem reduces to a combinatorial case.

In this first contribution, we restrict our analysis to symmetry classes strictly higher than
$[\DD_{2}]$, namely the orthotropic class. This choice is motivated by several considerations:
\begin{enumerate}
    \item Most engineering materials belong to symmetry classes that are equal to or higher than orthotropy; thus, this restriction focuses on cases of greatest practical relevance;
    \item The analysis of pair compatibility problems is relatively straightforward in the considered setting, and the partial result we propose can be demonstrated without the need to examine an exhaustive catalog of special cases.
\end{enumerate}
The resolution of the full problem is nevertheless of great importance and will be addressed in
a forthcoming contribution. Here, we prefer to emphasize the geometric nature and practical
relevance of the approach, rather than provide an exhaustive list that might obscure the most
significant results.


Our main result is the following one:
\begin{thm}[Number of exotic structures in $\Ela$ for symmetry class greater than orthotropic]\label{thm.NumExo}
Given the representation $\Ela$ of the rotation group $\SO(3)$, it comprises 18 exotic
structures in total for symmetry class higher than $[\DD_{2}]$. They are respectively 6 for
$[\DD_3]$, 6 for $[\DD_4]$, 6 for $[\OO(2)]$ and 0 for both $[\mathcal{O}]$ and $[\SO(3)]$.
\end{thm}
The corresponding complete structure of generic sets and exotic structures for overall symmetry
classes greater than $[\DD_{2}]$ are shown in \autoref{SymCla_Complet}. The detailed proofs are
provided in \autoref{S_Proof}. The labels with exponent \emph{g} representing for
\emph{generic} and \emph{e} for \emph{exotic}.
\begin{table}[htbp]
    \centering
 	\begin{tabular}{|c||c|c|c|c|c|c|c|c|}
 		\hline		
$[\G]$&$[\G_{\dT{h}^{a}}]$ &$[\G_{\dT{h}^{b}}]$ &$[\G_{\qT{H}}]$&$[\G_{(\dT{h}^{a},\dT{h}^{b})}]$&$[\G_{(\dT{h}^{a},\qT{H})}]$&$[\G_{(\dT{h}^{b},\qT{H})}]$&Label&Generic/Exotic \\
 		\hline
            \hline 
       
$[\SO(3)]$&\cellcolor{codegray}$[\SO(3)]$&\cellcolor{codegray}$[\SO(3)]$&\cellcolor{codegray}$\cellcolor{codegray}[\SO(3)]$&\cellcolor{codegray}$[\SO(3)]$&\cellcolor{codegray}$[\SO(3)]$& \cellcolor{codegray}$[\SO(3)]$&\cellcolor{codegray}$[\SO(3)]^{g}$&\cellcolor{codegray}Generic  \\
 		\hline

$[\mathcal{O}]$&\cellcolor{codegray}$[\SO(3)]$&\cellcolor{codegray}$[\SO(3)]$&\cellcolor{codegray}$[\mathcal{O}]$&\cellcolor{codegray}$[\SO(3)]$&\cellcolor{codegray}$[\mathcal{O}]$&\cellcolor{codegray}$[\mathcal{O}]$&\cellcolor{codegray}$[\mathcal{O}]^{g}$&\cellcolor{codegray}Generic  \\
 		\hline

 \multirow{7}{*}{\centering $[\OO(2)]$}& \cellcolor{codegray}$[\OO(2)]$&\cellcolor{codegray}$[\OO(2)]$ &\cellcolor{codegray}$[\OO(2)]$&\cellcolor{codegray}$[\OO(2)]$&\cellcolor{codegray}$[\OO(2)]$&\cellcolor{codegray}$[\OO(2)]$&\cellcolor{codegray}$[\OO(2)]^{g}$&\cellcolor{codegray}Generic \\
 \cline{2-9}
 &$[\SO(3)]$&$[\OO(2)]$ &$[\OO(2)]$&$[\OO(2)]$&$[\OO(2)]$&$[\OO(2)]$&$[\OO(2)]^{e}_{1}$&\multirow{6}{*}{\centering Exotic} \\
 \cline{2-8}
   &$[\OO(2)]$ &$[\SO(3)]$&$[\OO(2)]$&$[\OO(2)]$&$[\OO(2)]$&$[\OO(2)]$&$[\OO(2)]^{e}_{2}$& \\
 \cline{2-8}
 &$[\OO(2)]$ &$[\OO(2)]$ &$[\SO(3)]$&$[\OO(2)]$&$[\OO(2)]$&$[\OO(2)]$&$[\OO(2)]^{e}_{3}$& \\
 \cline{2-8}		
&$[\SO(3)]$ &$[\SO(3)]$&$[\OO(2)]$&$[\SO(3)]$&$[\OO(2)]$&$[\OO(2)]$&$[\OO(2)]^{e}_{4}$& \\
 \cline{2-8}
  &$[\SO(3)]$&$[\OO(2)]$&$[\SO(3)]$&$[\OO(2)]$&$[\SO(3)]$&$[\OO(2)]$&$[\OO(2)]^{e}_{5}$& \\
 \cline{2-8}
 &$[\OO(2)]$ &$[\SO(3)]$&$[\SO(3)]$&$[\OO(2)]$&$[\OO(2)]$&$[\SO(3)]$&$[\OO(2)]^{e}_{6}$& \\
 				\hline

 		\multirow{7}{*}{\centering $[\DD_4]$}& \cellcolor{codegray}$[\OO(2)]$ &\cellcolor{codegray}$[\OO(2)]$ &\cellcolor{codegray}$[\DD_4]$&\cellcolor{codegray}$[\OO(2)]$&\cellcolor{codegray}$[\DD_4]$&\cellcolor{codegray}$[\DD_4]$&\cellcolor{codegray}$[\DD_4]^{g}$&\cellcolor{codegray}Generic \\
 		\cline{2-9}		
&$[\SO(3)]$ &$[\OO(2)]$ & $[\DD_4]$&$[\OO(2)]$&$[\DD_4]$&$[\DD_4]$&$[\DD_4]^{e}_{1}$&\multirow{6}{*}{\centering Exotic} \\
\cline{2-8}	
&$[\OO(2)]$ &$[\SO(3)]$&$[\DD_4]$&$[\OO(2)]$&$[\DD_4]$&$[\DD_4]$&$[\DD_4]^{e}_{2}$& \\
\cline{2-8}	
&$[\OO(2)]$ &$[\OO(2)]$&$[\mathcal{O}]$&$[\OO(2)]$&$[\DD_4]$&$[\DD_4]$&$[\DD_4]^{e}_{3}$& \\
\cline{2-8}	
&$[\SO(3)]$ &$[\SO(3)]$ & $[\DD_4]$&$[\SO(3)]$&$[\DD_4]$&$[\DD_4]$&$[\DD_4]^{e}_{4}$&\\
\cline{2-8}	
&$[\SO(3)]$ &$[\OO(2)]$ & $[\mathcal{O}]$&$[\OO(2)]$&$[\mathcal{O}]$&$[\DD_4]$&$[\DD_4]^{e}_{5}$&\\
\cline{2-8}
&$[\OO(2)]$ &$[\SO(3)]$ & $[\mathcal{O}]$&$[\OO(2)]$&$[\DD_4]$&$[\mathcal{O}]$&$[\DD_4]^{e}_{6}$&\\
\hline	
           
 			\multirow{7}{*}{\centering $[\DD_3]$}&\cellcolor{codegray} $[\OO(2)]$
&\cellcolor{codegray}$[\OO(2)]$&\cellcolor{codegray}$[\DD_3]$&\cellcolor{codegray}$[\OO(2)]$&\cellcolor{codegray}$[\DD_3]$&\cellcolor{codegray}$[\DD_3]$&\cellcolor{codegray}$[\DD_3]^{g}$&\cellcolor{codegray}Generic\\	
\cline{2-9}
&$[\SO(3)]$ &$[\OO(2)]$ & $[\DD_3]$&$[\OO(2)]$&$[\DD_3]$&$[\DD_3]$&$[\DD_3]^{e}_{1}$&\multirow{6}{*}{\centering Exotic} \\
\cline{2-8}	
&$[\OO(2)]$ &$[\SO(3)]$&$[\DD_3]$&$[\OO(2)]$&$[\DD_3]$&$[\DD_3]$&$[\DD_3]^{e}_{2}$& \\
\cline{2-8}	
&$[\OO(2)]$ &$[\OO(2)]$&$[\mathcal{O}]$&$[\OO(2)]$&$[\DD_3]$&$[\DD_3]$&$[\DD_3]^{e}_{3}$& \\
\cline{2-8}	
&$[\SO(3)]$ &$[\SO(3)]$ & $[\DD_3]$&$[\SO(3)]$&$[\DD_3]$&$[\DD_3]$&$[\DD_3]^{e}_{4}$&\\
\cline{2-8}	
&$[\SO(3)]$ &$[\OO(2)]$ & $[\mathcal{O}]$&$[\OO(2)]$&$[\mathcal{O}]$&$[\DD_3]$&$[\DD_3]^{e}_{5}$&\\
\cline{2-8}
&$[\OO(2)]$ &$[\SO(3)]$ & $[\mathcal{O}]$&$[\OO(2)]$&$[\DD_3]$&$[\mathcal{O}]$&$[\DD_3]^{e}_{6}$&\\
\hline	

\end{tabular}\\
\caption{Complete symmetry class structures for $[\G]$ higher than $[\DD_2]$.}
\label{SymCla_Complet}
\end{table}
\FloatBarrier

The analysis of the table shows that, for the symmetry classes under consideration, the
degeneration mechanism is exclusively of Type I. That is, once the combinatorial structure
based on the nature of the elementary covariants is established, the symmetry classes of the
resulting pairs are fully determined — there is no remaining degree of freedom in this respect.
This is what makes the situation relatively simple under the current assumption, reducing the
problem to one of elementary combinatorics.

Let us analyse in more detail the physical meaning of the situations reported in the table. In
the definition of an exotic elastic material, we adopted the condition of being
\emph{hypersymmetric}, meaning that the material lies between two symmetry classes. A first,
straightforward observation is that no exotic hypersymmetric structure exists for the cubic
class. If such a structure did exist, it would have to lie between the cubic class $[\oct]$ and
the isotropic class $[\SO(3)]$. For the other symmetry classes considered, their corresponding
exotic structures are organized in the lattice diagrams in \autoref{fig:illustration}.
\begin{figure}[tbp]
\centering
\begin{subfigure}[b]{0.45\textwidth}
    \input{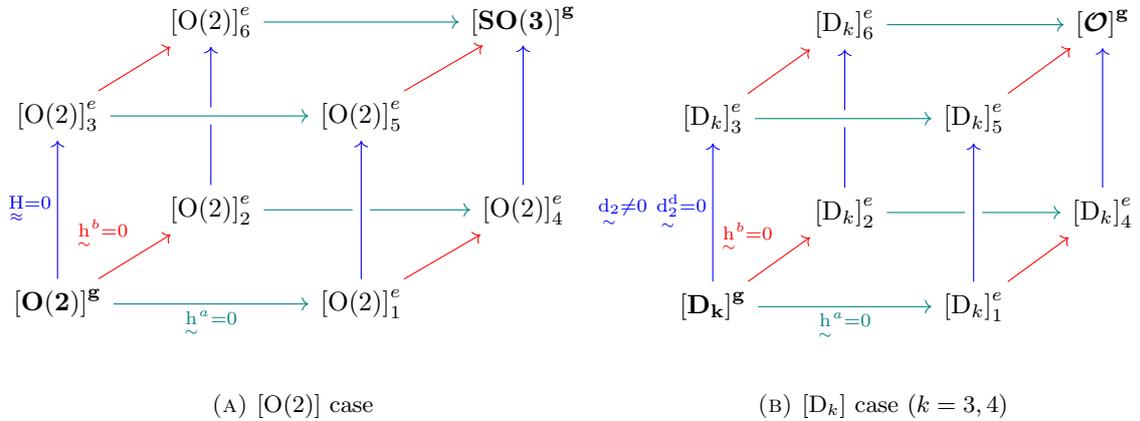}
\caption{$[\OO(2)]$ case}
\label{fig1:illustration}
\end{subfigure}
\begin{subfigure}[b]{0.45\textwidth}
\begin{tikzcd} [sep = .6 cm]
& \black{\left[\DD_k\right]^{e}_{6}} \arrow [rr, teal, crossing over] \arrow[from=dd,blue] & & \black{\mathbf{\left[\mathbfcal{O}\right]^{g}}}   \\
 \black{\left[\DD_k\right]^{e}_{3}}\arrow [rr, teal, crossing over]\arrow[ru,red] & & \black{\left[\DD_k\right]^{e}_{5}} \arrow[ru,red]\\
 & \black{\left[\DD_k\right]^{e}_{2}}   \arrow [rr, teal] & & \black{\left[\DD_k\right]^{e}_{4}} \arrow[uu,blue]\\
\black{\mathbf{\left[\mathbf{D}_{k}\right]^{g}}}\arrow [rr,"\dT{h}^{a}=0" below, teal] \arrow [ru,"\dT{h}^{b}=0"  ,red]\arrow[uu,"\dT{d_{2}}\neq0\ \dT{d_{2}^{d}}=0",blue] & & \black{\left[\DD_k\right]^{e}_{1}} \arrow [uu, crossing over, blue]\arrow[ru,red]\\
\end{tikzcd}
\caption{$[\DD_{k}]$ case ($k=3, 4$)}
\label{fig2:illustration} 
\end{subfigure}
\caption{Lattice of exotic structures}
\label{fig:illustration}
\end{figure}
\noindent Apart from the three covariants $\{{\dT{h}}^{a},{\dT{h}}^{b},\qT{H}\}$ of $\qT{C}$,
the covariant $\dT{d_{2}}:=\tr_{13}(\qT{H}:\qT{H})$ is also considered (the detailed covariants
construction has been computed in [\citep{olive2022}, Table2]) to characterize the symmetry
classes of elasticity tensors. $\dT{d_{2}^{d}}$ is defined as the deviatoric part of
$\dT{d_{2}}$.

In these diagrams, the generic structures—i.e., those obtained purely from symmetry
constraints—are highlighted in bold. The exotic cases, labelled from 1 to 6, occupy various
intermediate levels between the initial generic class and a higher generic symmetry class. In
these tables, it is clear that:
\begin{itemize}
\item Exotic structures of type $[\OO(2)]$ naturally lie between this class and the class $[\SO(3)]$;
\item Exotic structures of type $[\DD_k]$ for $k = 3, 4$ naturally lie between this class and the class $[\oct]$.
\end{itemize}
The diagram for exotic \red{sructures} $[\DD_k]$ with $k = 3, 4$ is based on the following result \cite{olive2022}:
\begin{thm}[Characterisation of the cubic class of $\qT{C}\in\Ela$]
Let $\qT{C}=f(\alpha,\beta,{\dT{h}}^{a},{\dT{h}}^{b},\qT{H})\in\Ela$ be an harmonic
decomposition of an elasticity tensor $\qT{C}$, where $\alpha, \beta$ are scalars,
${\dT{h}}^{a},{\dT{h}}^{b}\in\HH^{2}$ and $\qT{H}\in\HH^{4}$. Then,
\ben
\qT{C}\in\strata{\mathcal{O}}\  \text{if and only if}\  {\dT{h}^{a}}={\dT{h}}^{b}=\dT{d_{2}^{d}}=0\  \text{and } \dT{d_{2}}\neq0
\een
\end{thm}

\section{Particular cases of 3D Transverse Isotropic (TI) exotic materials}\label{S_ExoMaterials}

The results of the previous section remain abstract, as we merely enumerated the possible
exotic geometric structures for elasticity tensors with symmetry classes strictly higher than
orthotropy. To make these results effective—i.e., to give rise to actual exotic elastic
materials—one must choose an explicit harmonic decomposition. Such a choice corresponds to
selecting specific parametrisations for the $\HH^0$ and $\HH^2$ components, which in turn
allows for the mechanical interpretation of the exotic structures. As previously discussed,
such an explicit harmonic decomposition is not unique. We will focus here only on the two
decompositions introduced in \autoref{Sub_ExpDec}: the Clebsch–Gordan harmonic decomposition
(CGHD) and  the Schur–Weyl one (SWHD).

We will consider here three explicit examples of exotic Transverse Isotropic (TI) materials.
Each of these exotic cases is obtained by imposing an extra covariant condition on generic TI
case. To avoid \textcolor{red}{being} too abstract, for each of these examples, we will provide
an explicit matrix representation of an associated stiffness tensor. The matrix
representations will be given according to the so-called Kelvin convention\footnote{Also
referred to as the Mandel or Bechtcerew \citep{bechterew1926} convention, depending on the
author.}, that is, based on an orthonormal basis\footnote{In contrast to the classical Voigt
notation} adapted to $\mathbb{R}^6$ \citep{mehrabadi1990}. The normal form of generic TI tensor
is represented by:

\begin{equation}
\label{NF_TI}
\left[\qT{C}^{TI}\right]=
\begin{pmatrix}
C_{1111}&C_{1122} &C_{1133} & 0        & 0         & 0\\ 
   *    &C_{1111} &C_{1133} & 0        & 0         & 0\\ 
   *    &    *    &C_{3333} & 0        & 0         & 0 \\ 
   *    &    *    &    *    & 2C_{1313}& 0         &0 \\ 
   *    &    *    &    *    &   *      & 2C_{1313} &0  \\ 
   *    &    *    &    *    &   *      &   *       &C_{1111}-C_{1122}
\end{pmatrix}_{\mathcal{K}}.
\end{equation} 

\noindent As we can see, the tensor is characterized by five independent parameters. In each
exotic case, we will show how additional covariant conditions translate into constraints on
these parameters.

This section will also provide an opportunity to address an aspect that has not yet been
discussed: the stability under inversion. The question is whether an exotic stiffness tensor
necessarily leads to an exotic compliance tensor of the same type. As we shall see, exotic
structures are generally not preserved under inversion, and their stability depends on the type
of explicit harmonic decomposition considered.

\subsection{Uncoupled Transverse Isotropy (UTI)} 
\label{Sub_UTI}
The explicit CGHD interprets the elasticity tensor as a coupled law between the spherical and
deviatoric parts of the state tensors. For completeness and to maintain the logical flow of the
discussion, we restate \autoref{eq:ECGHD} below:
\ben
\begin{pmatrix}
    \dT{\sigma}^{d}\\ \dT{\sigma}^{s}
\end{pmatrix}=\begin{pmatrix}
\qT{H}+\dT{h}^{a}\boxtimes\dT{1}+\alpha \J & \frac{1}{3}\dT{h}^{b}\otimes\dT{1} \\ 
\frac{1}{3}\dT{1} \otimes \dT{h}^{b} & \beta\qT{K}
\end{pmatrix}\begin{pmatrix}
    \dT{\varepsilon}^{d}\\ \dT{\varepsilon}^{s}
\end{pmatrix}.
\een
With respect to this decomposition, the triplet $(\dT{h}^{a}, \dT{h}^{b}, \qT{H})$ is given the
following mechanical interpretation:
\begin{itemize}
\item $\dT{h}^{a}$ and $\qT{H}$ correspond to terms associated with deviatoric elasticity tensor $\qT{C}^{dd}$;
\item $\dT{h}^{b}$ corresponds to a term related to the coupling between spherical and deviatoric parts, that is, the tensors $\qT{C}^{ds}$ and $\qT{C}^{sd}$.
\end{itemize}

With respect to this decomposition, let consider to the exotic structure obtained by imposing
$\dT{h}^{b}=\dT{0}$, as represented in figure~\ref{fig:UTI}.

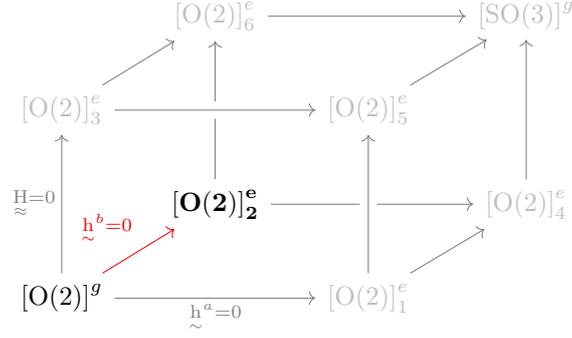
\begin{figure}[tbp]
\centering
\begin{tikzcd} [sep = .6 cm]
& \gray{\left[\OO(2)\right]^{e}_{6}}
\arrow [rr, gray, crossing over]
\arrow[from=dd,gray] & &
\gray{\left[\SO(3)\right]^{g}}
\\
\gray{\left[\OO(2)\right]^{e}_{3}}
\arrow [rr, gray, crossing over]\arrow[ru,gray] & & \gray{\left[\OO(2)\right]^{e}_{5}}
\arrow[ru,gray]
\\
& \black{\mathbf{\left[\mathbf{O}(2)\right]^{e}_{2}}}
\arrow [rr, gray] & &
\gray{\left[\OO(2)\right]^{e}_{4}}
\arrow[uu,gray]
\\
\black{\left[\OO(2)\right]^{g}}
\arrow [rr,"\dT{h}^{a}=0" below, gray]
\arrow [ru,"\dT{h}^{b}=0"  ,red]
\arrow[uu,"\qT{H}=0",gray] & &
\gray{\left[\OO(2)\right]^{e}_{1}}
\arrow [uu, crossing over, gray]
\arrow[ru,gray]
\\
\end{tikzcd}
\caption{Location of the UTI exotic structure (highlighted in bold) within the transition map}
\label{fig:UTI}
\end{figure}
\FloatBarrier
With respect to the CGHD, the vanishing of $\dT{h}^{b}$ is naturally interpreted as the
vanishing of the coupling between the spherical and deviatoric parts of the state tensors. The
elasticity is thus uncoupled:
\begin{equation}
\label{eq:ECGHD_UTI}
    \begin{pmatrix}
    \dT{\sigma}^{d}\\ \dT{\sigma}^{s}
\end{pmatrix}=\begin{pmatrix}
\qT{H}+\dT{h}^{a}\boxtimes\dT{1}+\alpha \J & 0 \\ 
0 & \beta\qT{K}
\end{pmatrix}\begin{pmatrix}
    \dT{\varepsilon}^{d}\\ \dT{\varepsilon}^{s}
\end{pmatrix}.
\end{equation}
\noindent It satisfies the first criterion in \textcolor{blue}{Definition~\ref{exotic_materials3D}}, as
the restriction arises not only from symmetry considerations but also from an additional
constraint, namely $\dT{h}^{b} = \dT{0}$. Moreover, such decoupling naturally arises in higher
symmetry classes, such as $[\mathcal{O}]$ and $[\SO(3)]$. Thus, this exotic elasticity can be
viewed as lying halfway between the classes $[\OO(2)]$ and $[\SO(3)]$, which reflects its
\emph{hypersymmetric} nature. We therefore classify it as an exotic elastic material, referred
to hereafter as Uncoupled Transverse Isotropy (UTI). Since the elasticity tensor of this exotic
material exhibits a block-diagonal matrix structure (see \autoref{eq:ECGHD_UTI}), the exotic
structure is stable under inversion. As a result, the associated compliance tensor also
corresponds to UTI.

The covariant condition $\dT{h}^{b} = \dT{0}$ can also be interpreted as a condition expressed
directly in terms of the components of the normal form of the tensor:
\begin{equation}
 \dT{h}^{b}=0   \quad \Leftrightarrow \quad C_{1111}+C_{1122}-C_{1133}-C_{3333}=0.
\end{equation}
By imposing this condition on the normal form of the generic transversely isotropic stiffness
tensor given in \autoref{NF_TI}, and substituting $C_{3333}$ with $C_{1111}+C_{1122}-C_{1133}$,
the normal form of UTI can be obtained:
\begin{equation}
\label{NF_UTI}
\left[\qT{C}^{UTI}\right]=
\begin{pmatrix}
C_{1111}&C_{1122} &C_{1133}                   & 0        & 0         & 0\\ 
   *    &C_{1111} &C_{1133}                   & 0        & 0         & 0\\ 
   *    &    *    &C_{1111}+C_{1122}-C_{1133} & 0        & 0         & 0 \\ 
   *    &    *    &    *                      & 2C_{1313}& 0         &0 \\ 
   *    &    *    &    *                      &   *      & 2C_{1313} &0  \\ 
   *    &    *    &    *                      &   *      &   *       &C_{1111}-C_{1122}
\end{pmatrix}_{\mathcal{K}}.
\end{equation}

\begin{exe}

To verify the decoupling between deviatoric and spherical elastic responses, consider the
identity tensor \( \dT{1} \), and calculate:
\[
\qT{C} : \dT{1} =  (C_{1111}+C_{1122} +C_{1133})  \, \dT{1}.
\]
Then, consider an harmonic tensor $\dT{h}$ and calculate
\begin{align*}
\left(\qT{C}:\dT{h}\right)_{11}&=(C_{1111}-C_{1133}) h_{11}+(C_{1122}-C_{1133})h_{22}, \\
\left(\qT{C}:\dT{h}\right)_{22}&=(C_{1122}-C_{1133}) h_{11}+(C_{1111}-C_{1133})h_{22} ,\\
\left(\qT{C}:\dT{h}\right)_{33}&=(2 C_{1133}-C_{1111}-C_{1122}) (h_{11}+h_{22}),
\end{align*}
from which it is found that $\tr(\qT{C}:\dT{h})=\dT{0}$.
Finally, an UTI stiffness tensor that fulfils the positive definiteness property is given
below:
\begin{equation}
\label{P3C7S3_E1}
\left[\qT{C}\right]=
\begin{pmatrix}
350     & 200     & 250     & 0   & 0  & 0\\ 
   *    & 350     & 250     & 0   & 0  & 0\\ 
   *    &    *    & 300     & 0   & 0  & 0 \\ 
   *    &    *    &    *    & 60  & 0  &0 \\ 
   *    &    *    &    *    &   * & 60 &0  \\ 
   *    &    *    &    *    &   * & *  &120
\end{pmatrix}_{\mathcal{K}},
\end{equation}
as its six eigenvalues $\{800, 150, 50, 60, 60, 120\}$ are strictly positive.
\end{exe}

\subsection{Transverse Isotropy with isotropic deviatoric elasticity (IDTI)} 
With respect to the CGHD decomposition, let \textcolor{blue}{us} consider now the exotic
structure obtained by imposing $\dT{h}^{a}=\dT{0}$ and $\qT{H}=\qT{0}$, as represented in
figure~\ref{fig:R0TI}.

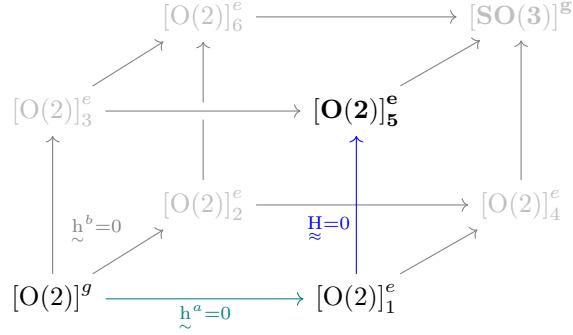
\begin{figure}[tbp]
\centering
\begin{tikzcd} [sep = .6 cm]
& \gray{\left[\OO(2)\right]^{e}_{6}} \arrow [rr, gray, crossing over] \arrow[from=dd,gray] & & \gray{\mathbf{\left[\mathbf{SO}(3)\right]^{g}}}   \\
 \gray{\left[\OO(2)\right]^{e}_{3}}\arrow [rr, gray, crossing over]\arrow[ru,gray] & & \black{\mathbf{\left[\mathbf{O}(2)\right]^{e}_{5}}} \arrow[ru,gray]\\
 & \gray{\left[\OO(2)\right]^{e}_{2}}   \arrow [rr, gray] & & \gray{\left[\OO(2)\right]^{e}_{4}} \arrow[uu,gray]\\
\black{\left[\OO(2)\right]^{g}}\arrow [rr,"\dT{h}^{a}=0" below, teal] \arrow [ru,"\dT{h}^{b}=0",gray]\arrow[uu,crossing over,gray] & & \black{\left[\OO(2)\right]^{e}_{1}} \arrow [uu, "\qT{H}=0"{pos=0.28, yshift=2pt}, blue]\arrow[ru,gray]\\
\end{tikzcd}
\caption{Location of the IDTI exotic structure (highlighted in bold) in the transition map}
\label{fig:R0TI}
\end{figure}

This point corresponds to the following Clebsch-Gordan block structure:
\begin{equation}
\label{eq:ROIT}
\begin{pmatrix}
    \dT{\sigma}^{d}\\ \dT{\sigma}^{s}
\end{pmatrix}=\begin{pmatrix}
\alpha \J &\frac{1}{3}\dT{h}^{b}\otimes\dT{1} \\ 
\frac{1}{3}\dT{1} \otimes \dT{h}^{b} & \beta\qT{K}
\end{pmatrix}\begin{pmatrix}
    \dT{\varepsilon}^{d}\\ \dT{\varepsilon}^{s}
\end{pmatrix},
\end{equation}
in which the vanishing of $\qT{H}$ and $\dT{h}^{a}$ leads to an isotropic deviatoric elasticity
tensor, $\qT{C}^{dd}$. Such exotic material will be termed Isotropic Deviatoric Transverse
Isotropy (IDTI). The resulting behaviour is clearly \emph{hypersymmetric} since isotropic
deviatoric elasticity only manifests generically for isotropic materials.

In \citet{Mou2023} within the $\RR^2$ setting,  we discussed an orthotropic material exhibiting isotropic deviatoric
elasticity, historically referred to as $R_0$-orthotropy \citep{Van02}. The exotic material
presented here can be regarded as a natural extension of this concept to $\mathbb{R}^3$.
Since the Clebsch–Gordan structure is not block-diagonal (cf. \autoref{eq:ROIT}), the exotic
property is not preserved under inversion, and the resulting compliance tensor is TI but not
IDTI.

The vanishing of $\qT{H}$ and $\dT{h}^{a}$ can also be interpreted using the components of the
normal form as follows:
\begin{equation*}
 \dT{h}^{a} =0  \quad \Leftrightarrow \quad {5 C_{1111} - 2 C_{3333} - 7 C_{1122} + 4 C_{1133}  - 6 C_{1313}}=0,
\end{equation*}
\begin{equation*}
 \qT{H}=0  \quad \Leftrightarrow \quad C_{1111}+C_{3333}-2 C_{1133} -4 C_{1313}=0.
\end{equation*}
Substituting these conditions into the normal form of the generic TI elasticity tensor given
in \autoref{NF_TI}, and replacing $C_{3333}$ and $C_{1313}$ with their expressions in terms of 
$C_{1111}$, $C_{1122}$ and $C_{1133}$, one obtains the normal form of IDTI:
\begin{equation}
\label{NF_R01}
\left[\qT{C}^{IDTI}\right]=
\begin{pmatrix}
C_{1111}&C_{1122}  &C_{1133}                     & 0                & 0                 & 0\\ 
   *    &C_{1111}  &C_{1133}                     & 0                & 0                 & 0\\ 
   *    &    *     &C_{1111}-2C_{1122}+2C_{1133} & 0                & 0                 & 0 \\ 
   *    &    *     &    *                        & C_{1111}-C_{1122}& 0                 &0 \\ 
   *    &    *     &    *                        &   *              & C_{1111}-C_{1122} &0  \\ 
   *    &    *     &    *                        &   *              &   *               &C_{1111}-C_{1122}
\end{pmatrix}_{\mathcal{K}}.
\end{equation} 

\begin{exe}
To verify the isotropic deviatoric elasticity, let's calculate the tensor $\qT{C}^{dd}$:
\begin{equation*}
\qT{C}^{dd} = (C_{1111}-C_{1122}) \; \J.
\end{equation*}
Finally, an IDTI elasticity tensor that fulfils the positive definiteness property is given
below:
\begin{equation}
[\qT{C}]=
\begin{pmatrix}
350     & 200     & 250     & 0   & 0  & 0\\ 
   *    & 350     & 250     & 0   & 0  & 0\\ 
   *    &    *    & 450     & 0   & 0  & 0 \\ 
   *    &    *    &    *    & 150  & 0  &0 \\ 
   *    &    *    &    *    &   * & 150 &0  \\ 
   *    &    *    &    *    &   * & *  &150
\end{pmatrix}_{\mathcal{K}},
\end{equation}
as its six eigenvalues $\{50(10+\sqrt{51}),50(10-\sqrt{51}),150,150,150,150\}$ are strictly
positive.
\end{exe}

\subsection{Transverse Isotropy with Isotropic Young's modulus (IYTI)}

In this final example, we focus on the directional Young's modulus, denoted by $E(\V{n})$, with
$\V{n} \in S^2$. As we shall see, the general expression of this function naturally involves
the Schur–Weyl decomposition of the compliance tensor $\qT{S}$, which is the inverse of the
stiffness tensor $\qT{C}$.

Formally, the Young's modulus in the direction $\V{n}$ is expressed as:
\ben
\frac{1}{E(\V{n})} = \frac{\varepsilon_{nn}}{\sigma_{nn}} = \qT{S} :: (\V{n} \otimes \V{n} \otimes \V{n} \otimes \V{n}).
\een
Since the tensor $\V{n} \otimes \V{n} \otimes \V{n} \otimes \V{n}$ is fully symmetric, the
expression of the Young's modulus depends only on the totally symmetric part of $\qT{S}$. 

The Schur–Weyl decomposition of the compliance tensor $\qT{S}$ yields:
\[
\qT{S} = \qT{S}^s + \qT{S}^a,\qquad
\text{with}\quad
\begin{cases}
\qT{S}^s = f(\beta^{-}, \dT{h}^{b-},\qT{H}^{-}) \\
\qT{S}^a = g(\alpha^{-}, \dT{h}^{a-})
\end{cases},
\]
In this expression above, the superscript $^{-}$ is used to indicate that the corresponding
term originates from the compliance tensor, rather than from the stiffness tensor. The detailed
expressions for the functions $f$ and $g$ can be found in \autoref{Sub_ExpDec}.

From this decomposition, we observe that the Young's modulus depends on $\qT{H}^{-}$ and
$\dT{h}^{b-}$, but not on $\dT{h}^{a-}$. One can therefore envision a situation in which both
$\qT{H}^-$ and $\dT{h}^b$ vanish, while $\dT{h}^{a-}$ is transverse isotropic. This would
result in a Transverse Isotropic elastic material with an Isotropic directional Young's modulus
(IYTI). Such a case is easily located within the lattice of exotic degeneracies of transverse
isotropy as presented in~\autoref{fig:IY-TI}.

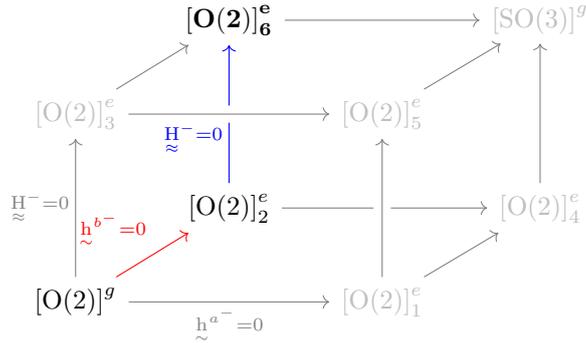
\begin{figure}[tbp]
\centering
\begin{tikzcd} [sep = .6 cm]
& \black{\mathbf{\left[\mathbf{O}(2)\right]^{e}_{6}}} \arrow [rr, gray, crossing over] \arrow[from=dd,"\qT{H}^{-}=0"{pos=0.28, yshift=2pt},blue] & & \gray{\left[\SO(3)\right]^{g}}   \\
 \gray{\left[\OO(2)\right]^{e}_{3}}\arrow [rr, gray, crossing over]\arrow[ru,gray] & & \gray{\left[\OO(2)\right]^{e}_{5}} \arrow[ru,gray]\\
 & \black{\left[\OO(2)\right]^{e}_{2}}   \arrow [rr, gray] & & \gray{\left[\OO(2)\right]^{e}_{4}} \arrow[uu,gray]\\
\black{\left[\OO(2)\right]^{g}}\arrow [rr,"\dT{h}^{a^{-}}=0" below, gray] \arrow [ru,"\dT{h}^{b^{-}}=0"  ,red]\arrow[uu,"\qT{H}^{-}=0",gray] & & \gray{\left[\OO(2)\right]^{e}_{1}} \arrow [uu, crossing over, gray]\arrow[ru,gray]\\
\end{tikzcd}
\caption{Location of the exotic structure (highlighted in bold) of IYTI within the transition map}
\label{fig:IY-TI}
\end{figure}
\FloatBarrier
The normal form of a TI compliance tensor holds the same forme as the stiffness tensor
presented in \autoref{NF_TI}. We recall here the elements of the decomposition computed as a
function of elasticity tensor given in Proposition \ref{prop2}, the vanishing of $\qT{H}^{-}$ and
$\dT{h}^{b-}$ can be interpreted in terms of the components of tensor $\qT{S}$ as follows:

\begin{equation*}
 \dT{h}^{b-} =\dT{0}  \quad \Leftrightarrow \quad  4 S_{1111} - 3 S_{3333} - S_{1133}  - 2 S_{1313}=0,
 \end{equation*}
 \begin{equation*}
 \qT{H}^{-}  =\qT{0}  \quad \Leftrightarrow \quad S_{1111} + S_{3333} - 2 S_{1133}  - 4 S_{1313}=0.
\end{equation*}
Replacing $S_{3333}$ and $S_{1313}$ with their expression in terms of $S_{1111}$ and
$S_{1133}$, a normal form of IYTI can be obtained:
\begin{equation}
\label{NF_IYTI1}
[\qT{S}]^{IYTI}=
\begin{pmatrix}
S_{1111}&S_{1122}  &S_{1133} & 0                & 0                 & 0\\ 
   *    &S_{1111}  &S_{1133} & 0                & 0                 & 0\\ 
   *    &    *     &S_{1111} & 0                & 0                 & 0 \\ 
   *    &    *     &    *    & S_{1111}-S_{1133}& 0                 &0 \\ 
   *    &    *     &    *    &   *              & S_{1111}-S_{1133} &0  \\ 
   *    &    *     &    *    &   *              &   *               &S_{1111}-S_{1122}
\end{pmatrix}_{\mathcal{K}}.
\end{equation} 
In this case, the Young’s modulus is isotropic and can be expressed as:
\ben
\dfrac{1}{E(\V{n})}=\beta^{-}=S_{1111}.
\een
The directional Young's modulus is generically isotropic only in the isotropic class, which
makes the present case clearly non-standard and exotic. This example is of particular interest
because:
\begin{enumerate}
    \item it has a fairly intuitive physical interpretation;
    \item it is traditionally accepted that an isotropic directional Young's modulus implies an isotropic elastic material.
\end{enumerate}
It should be noted, however, that this particular example of exotic elastic materials has
already been identified by several authors in the literature, notably by J.
\textcolor{blue}{Rychlewski} \citep{Ryc01} who was one of the pioneers in this field. The
contribution of Q.-C. He \citep{He2004} on this subject, around the same period, can also be
mentioned.

Regarding the question of the stability of the property under inversion, it is noted that the
Schur-Weyl decomposition does not preserve the linear operator structure of the elasticity
tensor. As a result, the property is not stable under inversion, and the symmetric part
$\qT{C}^s$ will not reduce to its isotropic component alone.

\begin{exe}
Finally, an IYTI elasticity tensor that fulfils the positive definiteness property is given
below:
\begin{equation}
[\qT{C}]=\frac{1}{9672}
\begin{pmatrix}
1183 & 377 & 468 & 0 & 0 & 0 \\
 377 & 1183 & 468 & 0 & 0 & 0 \\
 468 & 468 & 1248 & 0 & 0 & 0 \\
 0 & 0 & 0 & 744 & 0 & 0 \\
 0 & 0 & 0 & 0 & 744 & 0 \\
 0 & 0 & 0 & 0 & 0 & 806 \\
\end{pmatrix}_{\mathcal{K}}
\end{equation}
Its associated compliance is
\begin{equation}
[\qT{S}]=
\begin{pmatrix}
10      & -2       & -3     & 0   & 0  & 0\\ 
   *    & 10     &   -3   & 0   & 0  & 0\\ 
   *    &    *    & 10     & 0   & 0  & 0 \\ 
   *    &    *    &    *    & 13  & 0  &0 \\ 
   *    &    *    &    *    &   * & 13 &0  \\ 
   *    &    *    &    *    &   * & *  &12
\end{pmatrix}_{\mathcal{K}}
\end{equation}
which six eigenvalues $\{9+\sqrt{19},9-\sqrt{19},13,13,12,12\}$ are strictly positive.
\end{exe}

\section{Proof of the main results}
\label{S_Proof}

This final section is devoted to the proofs of the two main results of this paper:
\begin{enumerate}
    \item Proposition \ref{prop1}, concerning the generic geometric structure of the symmetry classes of $\Ela$; 
    \item Theorem \ref{thm.NumExo}, concerning the determination of all exotic geometric structures for symmetry classes above orthotropy.
\end{enumerate}

\subsection{Proof of Proposition \ref{prop1}}

Before presenting the proof of Proposition \ref{prop1}, we introduce the Curie principle, as it will be frequently referenced.
\begin{lem}[Curie principle]
Let a group $\mathrm{H}$ act on sets $X$ and $Y$. If a map $F : X \rightarrow Y$ is
$\mathrm{H}-$equivariant and if an element $x \in X$ is fixed under the action of $\mathrm{H}$,
then its image $y = F(x)$ is also fixed by $\mathrm{H}$.
\end{lem}
Let us denote by $\G_{x}$ the symmetry group of an element $x\in X$, it is the group defined as
\ben
\G_{x}:=\{g\in\mathrm{H},\ g\cdot x=x\}
\een
The Curie Principle means that
\ben
\G_{x}\subset\G_{F(x)}
\een
i.e the symmetry group of $x$ is a subgroup of the symmetry group of $F(x)$ which is
\textcolor{blue}{sometimes formulated} as \emph{the symmetry of the causes are included in the
symmetries of the consequences.}   In the present context, the Curie Principle implies that the
covariants of $\qT{C}$ can not have lower symmetry than $\qT{C}$ itself.

\textcolor{red}{Let us recall the following results concerning the symmetry classes of $\HH^2$
and $\HH^4$}
\begin{equation}\label{eq:ClasH}
\textcolor{red}{
\begin{split}
\SymC(\HH^{2})&=\{[\DD_2],[\OO(2)],[\SO(3)]\} \\
\SymC(\HH^{4})&=\{[1],[\ZZ_2],[\DD_2],[\DD_3],[\DD_4],[\OO(2)],[\mathcal{O}],[\SO(3)]\}
\end{split}
}
\end{equation}

The generic structures of the 8 symmetry classes, which are the subject of Proposition \ref{prop1},
 are determined through a two-step process for each class $[\G]$ of the resulting elasticity tensor:
\begin{enumerate}
    \item determination of the minimal symmetry class for each covariant compatible with $[\G]$. This step establishes the sets $\{[\G_{\dT{h}^{a}}],[\G_{\dT{h}^{b}}],[\G_{\qT{H}}]\}$; 
    \item determination of the least symmetric relative orientation between each pair of covariants compatible with $[\G]$. This step fixes the sets $\{[\G_{(\dT{h}^{a},\dT{h}^{b})}],[\G_{(\dT{h}^{a},\qT{H})}],[\G_{(\dT{h}^{b},\qT{H})}]\}$.
\end{enumerate}

The determination of the symmetry classes of pairs of covariants relies on the clips operations
introduced earlier in the paper; their definition and main properties are given in
\autoref{App_cliOpe}. The results of clips product between classes are detailed in Theorem
\ref{clipspro}. We now consider the different symmetry classes of $\Ela$, from the least to the
most symmetric:

\begin{itemize}
    \item \textbf{Triclinic $[1]$}\\
In the generic geometric  structure,  the symmetry class of each covariant should be minimal,
from  \autoref{eq:ClasH} it can be concluded that
$\{[\G_{\dT{h}^{a}}],[\G_{\dT{h}^{b}}],[\G_{\qT{H}}]\}=\{[\DD_2] ,[\DD_2], [1]\}$. In order for
the triplet to be fully generic, the covariants must occupy mutually generic positions, meaning
that:
$\{[\G_{(\dT{h}^{a},\dT{h}^{b})}],[\G_{(\dT{h}^{a},\qT{H})}],[\G_{(\dT{h}^{b},\qT{H})}]\}=\{[1],
[1],[1]\}$.
\item \textbf{Monoclinic $[\ZZ_2]$}\\
In generic structure the symmetry classes should be minimal. From the Curie Principle the
covariant of $\qT{C}$ can not be less symmetry than $[\ZZ_2]$,  as a result
$\{[\G_{\dT{h}^{a}}],[\G_{\dT{h}^{b}}],[\G_{\qT{H}}]\}=\{[\DD_2] ,[\DD_2], [\ZZ_2]\}$.
Regarding the relative orientation of the covariants, it must correspond to the minimal
configuration compatible with the class $[\ZZ_{2}]$. Based on the symmetry class analysis, we
have:
$\{[\G_{(\dT{h}^{a},\dT{h}^{b})}],[\G_{(\dT{h}^{a},\qT{H})}],[\G_{(\dT{h}^{b},\qT{H})}]\}=\{[\ZZ_2],
[\ZZ_2],[\ZZ_2]\}$.
\item \textbf{Orthotropic $[\DD_2]$}\\
The minimal symmetries of individual covariants compatible with the symmetry class $[\DD_2]$
are $\{[\G_{\dT{h}^{a}}],[\G_{\dT{h}^{b}}],[\G_{\qT{H}}]\}=\{[\DD_2] ,[\DD_2], [\DD_2]\}$.
Regarding the pair the minimal symmetries are as follows
$\{[\G_{(\dT{h}^{a},\dT{h}^{b})}],[\G_{(\dT{h}^{a},\qT{H})}],[\G_{(\dT{h}^{b},\qT{H})}]\}=\{[\DD_2],
[\DD_2],[\DD_2]\}$.
\item \textbf{Trigonal $[\DD_3]$}\\
For this case, the symmetry classes of $\dT{h}^{a}$ and $\dT{h}^{b}$ should not be less than
$[\DD_3]$, thus the least symmetric class compatible is $[\OO(2)]$. As a consequence
$\{[\G_{\dT{h}^{a}}],[\G_{\dT{h}^{b}}],[\G_{\qT{H}}]\}=\{[\OO(2)],[\OO(2)], [\DD_3]\}$. For the
classes of the pair, the minimal combination obtained from the clips operation is
$\{[\G_{(\dT{h}^{a},\dT{h}^{b})}],[\G_{(\dT{h}^{a},\qT{H})}],[\G_{(\dT{h}^{b},\qT{H})}]\}=\textcolor{red}{\{[\OO(2)],[\DD_3],[\DD_3]\}}$.
\item \textbf{Tetragonal $[\DD_4]$}\\
For this case, the symmetry classes of $\dT{h}^{a}$ and $\dT{h}^{b}$ remain as $[\OO(2)]$ and
that for $\qT{H}$ should be $[\DD_4]$. The corresponding minimal symmetry classes for the
coupled pairs are
$\{[\G_{(\dT{h}^{a},\dT{h}^{b})}],[\G_{(\dT{h}^{a},\qT{H})}],[\G_{(\dT{h}^{b},\qT{H})}]\}=\textcolor{red}{\{[\OO(2)],[\DD_4],[\DD_4]\}}$.
\item \textbf{Transverse Isotropy  $[\OO(2)]$}\\
The minimal symmetries of individual components compatible with the symmetry class $[\OO(2)]$
are $\{[\G_{\dT{h}^{a}}],[\G_{\dT{h}^{b}}],[\G_{\qT{H}}]\}=\{[\OO(2)] ,[\OO(2)], [\OO(2)]\}$.
The clips operations between covariants result in the following minimal symmetry of coupled
pairs
$\{[\G_{(\dT{h}^{a},\dT{h}^{b})}],[\G_{(\dT{h}^{a},\qT{H})}],[\G_{(\dT{h}^{b},\qT{H})}]\}=\{[\OO(2)],
[\OO(2)],[\OO(2)]\}$.
\item \textbf{Cubic $[\mathcal{O}]$}\\
The minimal symmetry of $\dT{h}^{a}$ and $\dT{h}^{b}$ that is not lower than $[\mathcal{O}]$ is
$[\SO(3)]$, and that for $\qT{H}$ is $[\mathcal{O}]$, thus we have
$\{[\G_{\dT{h}^{a}}],[\G_{\dT{h}^{b}}],[\G_{\qT{H}}]\}=\{[\SO(3)] ,[\SO(3)], [\mathcal{O}]\}$.
The clips operations between covariants result in the following minimal symmetry of coupled
pairs
$\{[\G_{(\dT{h}^{a},\dT{h}^{b})}],[\G_{(\dT{h}^{a},\qT{H})}],[\G_{(\dT{h}^{b},\qT{H})}]\}=\{[\SO(3)],
[\mathcal{O}],[\mathcal{O}]\}$.

    \item \textbf{Generic isotropic $[\SO(3)]$}\\
    This situation is direct, thanks to the Curie Principle, the covariant of $\qT{C}$ can not
    have less symmetry than $\qT{C}$ itself. Thus:
    $$\{[\G_{\dT{h}^{a}}],[\G_{\dT{h}^{b}}],[\G_{\qT{H}}],[\G_{(\dT{h}^{a},\dT{h}^{b})}],[\G_{(\dT{h}^{a},\qT{H})}],[\G_{(\dT{h}^{b},\qT{H})}]\}=\{[\SO(3)],[\SO(3)],
    [\SO(3)],[\SO(3)],[\SO(3)],[\SO(3)]\}$$
\end{itemize}

This completes the proof of Proposition \ref{prop1}.

\subsection{Proof of Theorem \ref{thm.NumExo}}
\label{Sub_Result}


In this paper, the classification procedure follows a two-step approach. First, a maximal set
of possible geometric structures is enumerated by applying the clips product, a purely
combinatorial operation. This preliminary step assumes that all components of the geometric
structure can be freely specified, disregarding any constraints from pairwise compatibility
that may exist between them. Then, the admissibility of the resulting candidates is analysed,
ultimately yielding the complete list of geometric structures. In short, we will successively
consider two lists of geometric structures.
\begin{itemize}
    \item \textbf{Maximal exotic structures:} full set of exotic candidates obtained purely through clips products without considering pairwise compatibility conditions;
     \item \textbf{Admissible exotic structures:} a refined subset of the maximal exotic set, incorporating pairwise compatibility conditions.
\end{itemize}

We begin by presenting the restricted clips operation, focusing on the clips operations between
symmetry classes that are strictly higher than the one under consideration.

\begin{defn}[Restricted clips operation]
For two families (finite or infinite) of symmetry classes $\mathcal{F}_{1}$ and
$\mathcal{F}_{2}$, we define the special clips operator, denoted by $\circledcirc^{[\G]}$:
\ben
\mathcal{F}_{1}\circledcirc^{[\G]} \mathcal{F}_{2}:=\mathop{\bigcup}\limits_{[\Hh_{1}]>[\G],[\Hh_{2}]>[\G] }  [\Hh_{1}]\circledcirc^{[\G]}  [\Hh_{2}] \quad \textrm{with} \ [\Hh_{1}]\in\mathcal{F}_{1},[\Hh_{2}]\in\mathcal{F}_{2},
\een
\end{defn}

Accordingly, the maximal list is derived through the application of restricted clips operations
to symmetry classes strictly greater than  $[\DD_2]$.
\red{In practice, this amounts to filtering the Table of Theorem \ref{clipspro} to retain only the classes higher than $[\DD_{2}]$.}
It results in the list of the 23 geometric structures provided in \autoref{SymCla_Complet}. Among these candidates, five generic
structures correspond to the symmetry classes: $\{[\DD_3],[\DD_4], [\OO(2)], [\oct], [\SO(3)]
\}$. The remaining 18 geometric structures thus constitute a maximal list of candidate exotic
structures that decompose as follows: 6 exotic candidates for $[\DD_3]$, 6 exotic candidates
for $[\DD_4]$, 6 exotic candidates for $[\OO(2)]$. The  admissibility of those elements must
now be examined. To this end, we begin by noting that the 18 distinct structures fall into two
distinct categories:
\begin{enumerate}
    \item \textbf{Configurations containing at least one isotropic individual component:}\\
    In these cases, the three-mode clips product simplifies to a two-mode clips product. As a
    result, the compatibility issue discussed earlier no longer arises. The outcome of the
    clips operation directly yields the admissible exotic structures. We observe that 16 out of
    the 18 configurations satisfy this condition and can therefore be considered admissible.
    \item \textbf{Configurations in which none of the individual components is isotropic:}\\
    Only two configurations fall into this category:
    
\begin{tabular}{|c||c|c|c|c|c|c|}
 		\hline		
$[\G]$&$[\G_{\dT{h}^{a}}]$ &$[\G_{\dT{h}^{b}}]$ &$[\G_{\qT{H}}]$&$[\G_{(\dT{h}^{a},\dT{h}^{b})}]$&$[\G_{(\dT{h}^{a},\qT{H})}]$&$[\G_{(\dT{h}^{b},\qT{H})}]$ \\
 		\hline
            \hline 
\centering $[\DD_4]$&$[\OO(2)]$&$[\OO(2)]$&$[\mathcal{O}]$&$[\OO(2)]$&$[\DD_4]$&$[\DD_4]$ \\
 				\hline	
 \centering $[\DD_3]$&$[\OO(2)]$ &$[\OO(2)]$&$[\mathcal{O}]$&$[\OO(2)]$&$[\DD_3]$&$[\DD_3]$ \\
 				\hline
 	\end{tabular}
    
For these, admissibility cannot be directly concluded from the clips product alone and must be
proven explicitly.
\end{enumerate}
The following theorem is addressed to proof admissibility of these two configurations:
\begin{lem}[$\mathbf{[\G_{\dT{h}^{a}}]=[\G_{\dT{h}^{b}}]=[\G_{(\dT{h}^{a},\dT{h}^{b})}]=[\OO(2)]\Rightarrow [\G_{(\dT{h}^{b},\qT{H})}]=[\G_{(\dT{h}^{a},\qT{H})}]}$] 
\label{theorem5}
    If $\dT{h^{a}}$, $\dT{h^{b}}$ and $(\dT{h^{a}}, \dT{h^{b}})$ are all TI  then the symmetry
    class of the pair $(\dT{h^{b}}, \qT{H})$ is identical to that of $(\dT{h^{a}}, \qT{H})$.
\end{lem}


\begin{proof}
If $\dT{h}^{a}$ and $\dT{h}^{b}$ are strictly transverse isotropic, there exists unit vectors
$\V{n}$ and $\V{m}$ such that
\ben
\dT{h}^{a}=\alpha (\V{n}\otimes\V{n})',\quad 
\dT{h}^{b}=\beta (\V{m}\otimes\V{m})',\ \alpha\beta\neq 0
\een
in which $\dT{T}'$ indicate the deviatoric part of $\dT{T}$. In such a case $\dT{h}^{a}$ is
transverse isotropic with respect to $\V{n}$, while $\dT{h}^{b}$ is transverse isotropic with
respect to $\V{m}$. For arbitrary relative orientation of $\V{n}$ and $\V{m}$ the symmetry
class of the pair $(\dT{h}^{a},\dT{h}^{b})$ is $[\ZZ_2]$. The only possibility for being
$[\OO(2)]$ is $\V{m}=\pm\V{n}$, as consequence
\ben
\dT{h}^{b}=\gamma\dT{h}^{a},\ \quad\gamma=\textcolor{red}{\frac{\beta}{\alpha}}\neq0
\een
The tensor $\dT{h}^{a}$ and $\dT{h}^{b}$ being proportional, the symmetry classes of the pairs
$(\dT{h^{b}}, \qT{H})$ and $(\dT{h^{a}}, \qT{H})$ are identical.
\end{proof}

\red{It should be observed that there exists relative orientation of $\dT{h}$ and $\qT{H}$ such as
\[
 [G_{(\dT{h}^a,\qT{H})}] \;=\; \bigl[\,\mathcal{O} \cap \OO(2)_{\V n}\,\bigr]
 \]
Taking $\qT{H} \in \HH^4$ to be the cubic harmonic tensor such as $[G_{\qT{H}}] = [\mathcal{O}]$ and intersecting this
cubic symmetry with the $\OO(2)$-symmetry of a fixed axis $\V{n}$ gives the sought configurations
\[
  [G_{(\dT{h}^a,\qT{H})}] \;=\; \bigl[\,\mathcal{O} \cap \OO(2)_{\V n}\,\bigr]
  \;=\;
  \begin{cases}
    [\DD_4], & \V{n}\ \text{a $4$-fold (face) axis of the cube},\\[2pt]
    [\DD_3], & \V{n}\ \text{a $3$-fold (diagonal) axis of the cube},
  \end{cases}
\]
which thus proves admissibility of the two configurations.}


Thus it can be concluded that all these 18 exotic candidates are admissible. Their degeneracy
stems exclusively from the elementary covariants.

\section{Conclusions and perspectives}

In this work, extending the two–dimensional study initiated in \citep{Mou2023}, we investigated
exotic elastic materials in three dimensions. The main results are as follows:

\begin{itemize}
    \item \textbf{Determination of exotic structures for 3D linear elasticity with symmetry classes higher than orthotropic}. The geometrical tools have been applied to 3D linear elasticity, and based on our proposed definition of \emph{exotic materials}, it can be concluded that there are 18 exotic structures in 3D linear elasticity with symmetry classes higher than orthotropic. 
    
    \item \textbf{Presentation of some exotic materials.} 
    The aforementioned approach is applied to three instances of exotic elasticity: Uncoupled
    Transverse Isotropy (UTI), Transverse Isotropy with Isotropic Deviatoric elasticity
    (IDTI)\textcolor{red}{ and Transverse Isotropy} with Isotropic Young's modulus (IYTI). This
    highlights the robustness and applicability of our approach, and it provides theoretical
    support for further investigations into a wider array of exotic materials.
\end{itemize}

Compared with the 2D setting, the problem becomes significantly more intricate: harmonic
parametrizations are no longer unique and the number of possible degeneracies increases
dramatically. Rather than being a drawback, this complexity highlights the wide range of
possibilities enabled by the design of tailored elasticity tensors. When combined with topology
optimization and additive manufacturing, this approach opens broad practical opportunities and
is of significant interest for the design of architected materials.

Several natural extensions of this work can be identified.
\begin{itemize}
\item  From a theoretical perspective, exotic combinations associated with low–symmetry elasticity tensors remain to be explored, in particular for orthotropic, monoclinic and triclinic classes. The combinatorial structure of these cases is considerably more involved than for the situations considered in the present work, and the number of admissible configurations is expected to reach several hundreds, if not thousands.
\item  Topology optimization. Beyond existence results, identifying mesostructure geometries capable of realizing exotic elastic responses remains a key challenge. While this task has been addressed in two dimensions, it still needs to be carried out in three dimensions. The covariant formalism implicitly used throughout this work provides a promising framework for constructing cost functions suited to this class of optimization problems.
\end{itemize}

\begin{appendix}

\section{Clips operation}
\label{App_cliOpe}
The whole vector space $\Vv$ can be decomposed into a direct sum of irreducible $\Vv_{i}$ $(i=1,\cdots,N)$.
To obtain the symmetry classes of $\Vv$, we must therefore compute the symmetry classes of a direct sum $\Vv_{1}\oplus \cdots \oplus \Vv_{N}$,
knowing independently the symmetry classes of each space $\Vv_{i}$. 
To do so, we introduce in this \red{appendix some basic definitions associated with \emph{clips operation}}~\citep{Olive2019}. \\



\begin{defn}[clips operation]
For each symmetry class $[\Hh_{1}]$ and $[\Hh_{2}]$, we define the clip operator of $[\Hh_{1}]$ and $[\Hh_{2}]$, denoted by $[\Hh_{1}]\circledcirc[\Hh_{2}]$:
\ben
[\Hh_{1}]\circledcirc[\Hh_{2}]:=\left \{ \left [ \Hh_{1}\cap \vg\Hh_{2}\vg^{-1} \right ], \forall \vg\in \G \right \}
\een
\end{defn}

This definition immediately extends to two families (finite or infinite) $\mathcal{F}_{1}$ and $\mathcal{F}_{2}$ of symmetry classes:
\ben
\mathcal{F}_{1}\circledcirc \mathcal{F}_{2}:=\mathop{\bigcup}\limits_{[\Hh_{1}]\in\mathcal{F}_{1},[\Hh_{2}]\in\mathcal{F}_{2}}  [\Hh_{1}]\circledcirc [\Hh_{2}]
\een
This clips operation defines thus an operation on the set of symmetry classes $\SymC$ which is associative and commutative. We have moreover

\begin{center}
$[1]\circledcirc[\Hh]={[1]}$ and $[\G]\circledcirc[\Hh]={[\Hh]}$
\end{center}
for every symmetry class $[\Hh]$, where $1:=\{\mathbf{e}\}$ and $\mathbf{e}$ is the identity element of $\G$.\\

The symmetry classes of a direct sum of subspaces are obtained by the clips of their respective symmetry classes.

\begin{lem}
\label{P1C2S6_l4}
~\citep{Olive2019}
Let $(\rho,\Vv_{1})$ and $(\rho,\Vv_{2})$ be two linear representations of $\G$. Then
\ben
\SymC(\Vv_{1}\oplus \Vv_{2})=\SymC(\Vv_{1})\circledcirc \SymC(\Vv_{2})
\een
\end{lem}

The clips operations between SO(3)-closed subgroups is pivotal, and to make verification straightforward for readers, we have placed the full treatment in the \autoref{clipspro} rather than in this appendix.

\end{appendix}
\newpage
\bibliographystyle{elsarticle-harv}
\bibliography{References}
\end{document}